\documentclass{article} % For LaTeX2e
\usepackage{iclr2024_conference,times}

\usepackage{amsmath,amsfonts,bm}

\def\eqref#1{equation~(\ref{#1})}
\def\Eqref#1{Equation~(\ref{#1})}
\def\1{\bm{1}}

\DeclareMathAlphabet{\mathsfit}{\encodingdefault}{\sfdefault}{m}{sl}
\SetMathAlphabet{\mathsfit}{bold}{\encodingdefault}{\sfdefault}{bx}{n}

\newcommand{\ellone}{$\mathcal{L}_1$ }
\usepackage{soul}

\newcommand{\lmbrl}{$\mathcal{L}_1$-MBRL }
\usepackage{wrapfig}
\usepackage{cancel}
\newcommand{\xh}{\hat{x}}

\newcommand{\ua}{u_a}
\newcommand{\sm}{\sigma_m}
\newcommand{\summ}{\sigma_{um}}
\usepackage{amsmath}
\usepackage{amssymb}
\usepackage{amsthm}
\newtheorem*{theorem*}{Theorem}
\usepackage{bbm}
\usepackage{natbib}
\usepackage{graphicx}
\usepackage[ruled,vlined]{algorithm2e}
\usepackage{url}

\usepackage{booktabs}
\usepackage{color,colortbl}

\newtheorem{remark}{Remark}
\newtheorem{theorem}{Theorem}

\newtheorem{assumption}{Assumption}
\usepackage{bigints}
\definecolor{applegreen}{rgb}{0.0, 0.47, 0.44}
\newcommand{\norm}[1]{\| #1 \|}
\newcommand{\xs}{x^\star}
\newcommand{\us}{u^\star}

\newcommand{\smh}{\hat{\sigma}_m}
\newcommand{\summh}{\hat{\sigma}_{um}}
\usepackage{hyperref}
\usepackage{url}

\title{Robust  Model Based Reinforcement Learning \centering Using \ellone Adaptive Control}

\author{Minjun Sung\thanks{The authors equally contributed to this work.},~~Sambhu  H.~Karumanchi\footnotemark[1],~~Aditya Gahlawat,~~Naira Hovakimyan\\
Department of Mechanical Science \& Engineering\\
University of Illinois Urbana-Champaign\\
Urbana, IL, 61801, USA \\
\texttt{\{mjsung2,shk9,gahlawat,nhovakim\}@illinois.edu}
}

\iclrfinalcopy % Uncomment for camera-ready version, but NOT for submission.
\begin{document}

\maketitle
\begin{abstract}
We introduce $\mathcal{L}_1$-MBRL, a control-theoretic augmentation scheme for Model-Based Reinforcement Learning (MBRL) algorithms. Unlike model-free approaches, MBRL algorithms learn a model of the transition function using data and use it to design a control input. Our approach generates a series of approximate control-affine models of the learned transition function according to the proposed \textit{switching law}. Using the approximate model, control input produced by the underlying MBRL is perturbed by the \ellone adaptive control, which is designed to enhance the robustness of the system against uncertainties. Importantly, this approach is agnostic to the choice of MBRL algorithm,  enabling the use of the scheme with various MBRL algorithms. MBRL algorithms with \ellone augmentation exhibit enhanced performance and sample efficiency across multiple MuJoCo environments, outperforming the original MBRL algorithms, both with and without system noise.
\end{abstract}

\section{Introduction}\label{sec:intro}

Reinforcement learning (RL) combines stochastic optimal control with data-driven learning. Recent progress in deep neural networks (NNs) has enabled RL algorithms to make decisions in complex and dynamic environments~\citep{wang2022deep}. Reinforcement learning algorithms have achieved remarkable performance in a wide range of applications, including robotics~\citep{ibarz2021train,nguyen2019review}, natural language processing~\citep{wang2018video,wu2018study}, autonomous driving~\citep{milz2018visual,li2020deep}, and computer vision~\citep{yun2017action,wu2017learning}.

There are two main approaches to reinforcement learning: Model-Free RL (MFRL) and MBRL. MFRL algorithms directly learn a policy to maximize cumulative reward from data, while MBRL algorithms first learn a model of the transition function and then use it to obtain optimal policies~\citep{moerland2023model}. While MFRL algorithms have demonstrated impressive asymptotic performance, they often suffer from poor sample complexity~\citep{mnih2015human,lillicrap2016continuous,schulman2017proximal}. On the other hand, MBRL algorithms offer superior sample complexity and are agnostic to tasks or rewards~\citep{kocijan2004gaussian,deisenroth2013gaussian}. While MBRL algorithms have traditionally lagged behind MFRL algorithms in terms of asymptotic performance, recent approaches, such as the one presented by \citet{chua2018deep}, aim to bridge this gap.

In MBRL, learning a model of the transition function can introduce model (or epistemic) uncertainties due to the lack of sufficient data or data with insufficient information. Moreover, real-world systems are also subject to inherently random aleatoric uncertainties. As a result, unless sufficient data—both in quantity and quality—is available, the learned policies will exhibit a gap between expected and actual performance, commonly referred to as the sim-to-real gap~\citep{zhao2020sim}.

The field of robust and adaptive control theory has a rich history and was born out of a need to design a controller to address the uncertainties discussed above. Given that both MBRL algorithms and classical control tools depend on models of the transition function, it is natural to consider exploring the consolidation of robust and adaptive control with MBRL. However, such a consolidation is far from straightforward, primarily due to the difference between the class of models for which the robustness is considered. To analyze systems and design controllers for such systems, conventional control methods often assume extensively modeled dynamics that are gain scheduled, linear, control affine, and/or true up to parametric uncertainties~\citep{neal2004design,nichols1993gain}. On the other hand, MBRL algorithms frequently update highly nonlinear models (e.g. NNs) to enhance their predictive accuracy. The combination of this iterative updating and the model’s high nonlinearity creates a unique challenge in embedding robust and adaptive controllers within MBRL algorithms.
% MBRL algorithms often use highly nonlinear models (often NNs) that do not have true parameters corresponding to the ground truth dynamics, only optimal from a predictive sense, which makes consolidating MBRL with control theoretic tools challenging.

\subsection{Statement of contributions} We propose the \lmbrl framework as an add-on scheme to augment MBRL algorithms, which offers improved robustness against epistemic and aleatoric uncertainties. We \emph{affinize} the learned model in the control space according to the \emph{switching law} to construct a control-affine model based on which the \ellone control input is designed. The \emph{switching law} design provides a distinct capability to explicitly control the predictive performance bound of the \emph{state predictor} within the \ellone adaptive control architecture while harnessing the robustness advantages offered by the \ellone adaptive control. The \ellone add-on does not require any modifications to the underlying MBRL algorithm, making it agnostic to the choice of the baseline MBRL algorithm. To evaluate the effectiveness of the \lmbrl scheme, we conduct extensive numerical simulations using two baseline MBRL algorithms across multiple environments, including scenarios with action or observation noise. The results unequivocally demonstrate that the \lmbrl scheme enhances the performance of the underlying MBRL algorithms without any redesign or retuning of the \ellone controller from one scenario to another.

\subsection{Related Work} 
{\bf Control Augmentation of RL policies} is of significant relevance to our research. Notable recent studies in this area, including~\citet{cheng2022improving} and~\citet{arevalo2021model}, have investigated the augmentation of adaptive controllers to policies learned through MFRL algorithms. However, these approaches are limited by their assumption of known nominal models and their restriction to control-affine or nonlinear models with known basis functions, which restricts their applicability to specific system types. In contrast, our approach does not assume any specific structure or knowledge of the nominal dynamics. We instead provide a general framework to augment an \ellone  adaptive controller to the learned policy, while simultaneously learning the transition function.

{\bf Robust and adversarial RL methods} aim to enhance the robustness of RL policies by utilizing minimax optimization with adversarial perturbation, as seen in various studies~\citep{tobin2017domain,peng2018sim,loquercio2019deep,pinto2017robust}. However, existing methods often involve modifications to data or dynamics in order to handle worst-case scenarios, leading to poor general performance. In contrast, our method offers a distinct advantage by enhancing robustness without perturbing the underlying MBRL algorithm. This allows us to improve the robustness of the baseline algorithm without sacrificing its general performance.

{\bf Meta-(MB)RL} methods train models across multiple tasks to facilitate rapid adaptation to dynamic variations as proposed in~\citep{finn2017model,nagabandilearning,nagabandideep}. These approaches employ hierarchical latent variable models to handle non-stationary environments. However, they lack explicit provisions for uncertainty estimation or rejection, which can result in significant performance degradation when faced with uncertain conditions~\citep{chen2021context}. In contrast, the \lmbrl framework is purposefully designed to address this limitation through uncertainty estimation and explicit rejection.
% While both methods share the common goal to make rapid adaptation to its dynamic environments, they differ in their specific focus and mechanisms, with \lmbrl specifically designed to estimate and reject uncertainty. 
Importantly, our \lmbrl method offers the potential for effective integration with meta-RL approaches, allowing for the joint leveraging of both methods to achieve both environment adaptation and uncertainty rejection in a non-stationary and noisy environment.

\section{Preliminaries}\label{sec:prelim}

In this section we provide a brief overview of the two main components of our \lmbrl framework: MBRL and \ellone adaptive control methodology.
%%%%%%%%%%%%%%
\subsection{Model Based Reinforcement Learning}\label{subsec:MBRL_prelim}

This paper assumes a discrete-time finite-horizon Markov Decision Process (MDP), defined by the tuple $\mathcal{M} = (\mathcal{X},\mathcal{U}, f, r, \rho_0,\gamma, H)$. Here $\mathcal{X}\subset \mathbb{R}^n$ is the compact state space, $\mathcal{U}\subset \mathbb{R}^m$ is the compact action space, $f:\mathcal{X}\times \mathcal{U} \rightarrow \mathcal{X} $ is the deterministic transition function, $r:\mathcal{X}\times \mathcal{A}\rightarrow \mathbb{R}$ is a bounded reward function. Let $\xi(\mathcal{X})$ be the set of probability distributions over $\mathcal{X}$ and  $\rho_0 \in \xi(\mathcal{X})$ be the initial state distribution. $\gamma$ is the discount factor and $H \in \mathbb{N}$ is a known horizon of the problem. For any time step $t< H$, if $x_t \notin \mathcal{X}$ or $u_t \notin \mathcal{U}$, then the episode is considered to have failed such that $r(x_{t'},u_{t'}) = 0$ for all $t'= t, t+1, \ldots, H$. A policy is denoted by $\pi$ and is specified as $\pi := [\pi_1, \ldots, \pi_{H-1}]$, where  $\pi_t: \mathcal{X} \rightarrow \xi (\mathcal{A})$ and $\xi (\mathcal{A})$ is the set of probability distributions over $\mathcal{A}$. The goal of RL is to find a policy that maximizes the expected sum of the reward along a trajectory $\tau := (x_0,u_0,\cdots,x_{H-1},u_{H-1},x_H)$, or formally, to maximize $J(\pi) = \mathbb{E}_{x_0 \sim \rho_0, u_t\sim \pi_t}[\sum_{t=1}^H \gamma^t r(x_t,u_t)]$, where $x_{t+1}-x_t = f(x_t,u_t)$~\citep{nagabandi2018neural}. The trained model can be utilized in various ways to obtain the policy, as detailed in Sec.~\ref{MBRL}.

% treated as a simulator from which fictitious samples can be collected to train the policy $\pi_\phi$, often represented by a NN. One may not necessarily define a parameterized policy network and directly search for the optimal control by using methods like model predictive control (MPC). However, such methods can also be distilled into a policy network as done in ~\citet{levine2014learning,levine2013guided,zhang2019solar}, and therefore we can, in general, consider a standard MBRL method to include model learning and an optimization of a parameterized policy. 
%%%%%%%%%%%%%%
\subsection{\ellone Adaptive Control}\label{subsec:L1-prelim}

% As explained in Sec.~\ref{sec:intro} and Sec.~\ref{subsec:MBRL_prelim}, the MBRL algorithm learns a model of the transition function using which it computes an optimal policy. This model is subject to both epistemic and aleatoric uncertainties. 

The \ellone adaptive control theory provides a framework to counter the uncertainties with guaranteed transient and steady-state performance, alongside robustness margins~\citep{hovakimyan2010L1}. The performance and reliability of the \ellone adaptive control has been extensively tested on systems including robotic platforms~\citep{cheng2022improving,pravitra2020,wu20221}, NASA AirSTAR sub-scale aircraft~\citep{gregory2009l1,gregory2010flight}, and Learjet~\citep{ackerman2017evaluation}. While we give a brief description of the \ellone adaptive control in this subsection, we refer the interested reader to Appendix~\ref{sec:ellone-extended} for detailed explanation. 

Assume that the continuous-time dynamics of a system can be represented as
%%%%%%%%%%%%%%%
\begin{equation}\label{eqn:true_continuous}
    \dot{x}(t) = g(x(t)) + h(x(t))u(t) + d(t,x(t), u(t)), \quad x(0) = x_0,
\end{equation}
%%%%%%%%%%%%%%%
where $x(t) \in \mathbb{R}^n$ is the system state, $u(t) \in \mathbb{R}^m$ is the control input, $g:\mathbb{R}^n \rightarrow \mathbb{R}^n$ and $h:\mathbb{R}^n \rightarrow \mathbb{R}^{n \times m}$ are known nonlinear functions, and $d(t,x(t), u(t)) \in \mathbb{R}^n$ represents the unknown residual containing both the model uncertainties and the disturbances affecting the system. 

%%%%%%%%%%%%%%%%%%%%
% \begin{remark}
%     We are considering continuous in time control-affine dynamics for the brief exposition on \ellone adaptive control. However, as mentioned before, the MBRL algorithm learns a fully nonlinear and discrete-time transition function. Our approach on how to consolidate \ellone and MBRL algorithms is a feature of our solution and thus the model provided in~\Eqref{eqn:true_continuous} should not detract from the applicability of the \lmbrl methodology. {\bf this last sentence is not clear, what is the point here? We have also L1 for non-affine systems like $\dot x=f(t,x,u)$; Ronald Choe published such paper in IEEE TAC. Is there anything more nonlinear? 	R. Choe, E. Xargay, N. Hovakimyan, L1  Adaptive Control for a Class of Uncertain Nonaffine-in-Control Nonlinear Systems, IEEE Transactions on Automatic Control, vol. 61, No.3, pp. 840-846, 2016.}
% \end{remark}
%%%%%%%%%%%%%%%%%%%%
% Since any high-level optimal controller does not have access to the residual $d(t,x(t))$, it produces a \emph{baseline} input $\us$ that produces a \emph{nominal} state trajectory $\xs$ whose evolution is governed by {\bf what does this sentence mean? no control can be assumed to have access to the uncertainty. why all of a sudden this optimal control falls off from the sky so suddenly?}
Consider a desired control input $\us(t) \in \mathbb{R}^m$ and the induced desired state trajectory $\xs(t) \in \mathbb{R}^n$ based on the nominal (uncertainty-free) dynamics  
%%%%%%%%%%%%%%%
\begin{equation}\label{eqn:nominal_continuous}
\dot{x}^\star(t) = g(\xs(t))+h(\xs(t))\us(t), \quad \xs(0) = x_0.
\end{equation}
%%%%%%%%%%%%%%%
If we directly apply $\us(t)$ to the true system in~\Eqref{eqn:true_continuous}, the presence of the uncertainty $d(t, x(t), u(t))$ can cause the actual trajectory to diverge unquantifiably from the nominal trajectory. 
The \ellone adaptive controller computes an additive control input $u_a(t)$ to ensure that the augmented input $u(t) = \us(t) + u_a(t)$ keeps the actual trajectory $x(t)$ bounded around the nominal trajectory $\xs(t)$ in a quantifiable and uniform manner.

The \ellone adaptive controller has three components: the state predictor, the adaptive law, and a low-pass filter. The state predictor is given by
%%%%%%%%%%%%%%%%%
\begin{equation}\label{eqn:predictor_continuous}
    \dot{\xh}(t) = g(x(t))+h(x(t))(\us(t) + \ua(t)) + \hat{\sigma}(t) + A_s \Tilde{x}(t),
\end{equation}
%%%%%%%%%%%%%%%%%%%%%%%%%%
with the initial condition $\hat{x}(0)= \hat{x}_0$, where $\xh(t) \in \mathbb{R}^n$ is the state of the predictor, $\hat{\sigma}(t)$ is the estimate of $d(t,x(t), u(t))$,$ \:\Tilde{x}(t) = \xh (t)- x(t)$ is the state prediction error, and $A_s\in \mathbb{R}^{n\times n}$ is a Hurwitz matrix chosen by the user. Furthermore, $\hat{\sigma}(t)$ can be decomposed as 
\begin{equation}\label{eqn:d(t,x(t))}
  \hat{\sigma}(t)   = h(x(t))\hat{\sigma}_m(t)   + h^\perp(x(t))\hat{\sigma}_{um}(t),    
\end{equation}
where $\hat{\sigma}_m(t)$ and $\hat{\sigma}_{um}(t)$ are the estimates of the \emph{matched} and \emph{unmatched} uncertainties. Here, $h^\perp(x(t)) \in \mathbb{R}^{n \times (n-m)}$ is a matrix satisfying $ h(x(t))^\top h^\perp(x(t)) = 0$ and $\text{rank}\left(\begin{bmatrix}h(x(t)),~h^\perp(x(t))  \end{bmatrix} \right)= n$. The existence of $h^\perp(x(t))$ is guaranteed, given $h(x(t))$ is a full-rank matrix. The role of the predictor is to produce the state estimate $\xh(t)$ induced by the uncertainty estimate $\hat{\sigma}(t)$. 

The uncertainty estimate is updated using the piecewise constant adaptive law given by
%%%%%%%%%%%%%%%%%%%%%%%%%%
\begin{equation}\label{eqn:adaptive_law_continuous}
    \hat{\sigma}(t) =\hat{\sigma}(iT_s) = -\Phi^{-1}(T_s)\mu(iT_s),~\quad t \in [iT_s,(i+1)T_s),~\quad i \in \mathbb{N},~\quad \hat{\sigma}(0) = \hat{\sigma}_0,
\end{equation}
%%%%%%%%%%%%%%%%%%%%%%%%%%
where $\Phi(T_s) = A_s^{-1} \left(\text{exp}(A_s T_s) - \mathbb{I}_n\right)$, $\mu(iT_s) = \text{exp}(A_s T_s) \tilde{x}(iT_s)$, and $T_s$ is the sampling time.

Finally, the control input is given by
%%%%%%%%%%%%%%%%%%%%%%%%%%
\begin{equation}\label{eqn:control_law_continuous}
    u(t) = \us(t) + \ua(t), \quad u_a(s) = -C(s)\mathfrak{L}\left[\hat{\sigma}_m(t) \right],
\end{equation}
%%%%%%%%%%%%%%%%%%%%%%%%%%
where $\mathfrak{L}[\cdot]$ denotes the Laplace transform, and the \ellone input $\ua(t)$ is the output of the low-pass filter $C(s)$ in response to the estimate of the matched component $\hat{\sigma}_m(t)$. The bandwidth of the low-pass filter is chosen to satisfy the small-gain stability condition~\citep{wang2011}. 
% The role of the low-pass filter $C(s)$ is to prevent the leaking of high-frequency signals into the system and thus destroying its robustness. Thus, $C(s)$ decouples the estimation loop from the control loop and therefore allows \emph{arbitrarily fast} computation of the adaptive estimates in~\Eqref{eqn:adaptive_law_continuous} by choosing a small $T_s > 0$ limited only by the computational hardware, while maintaining desired robustness margins~\cite{hovakimyan2010L1}.

\section{The \lmbrl Algorithm}\label{sec:l1mbrl}

In this section, we present the \lmbrl algorithm, illustrated in Fig.~\ref{fig:Framework}. We first explain a standard MBRL algorithm and describe our method to integrate \ellone adaptive control with it. 

%%%%%%%%%%%%%%%%%%%%%
\begin{wrapfigure}{r}{0.5\textwidth}
\vspace{-20pt}
\begin{center}
    \centering
    \includegraphics[width = 0.5\textwidth]{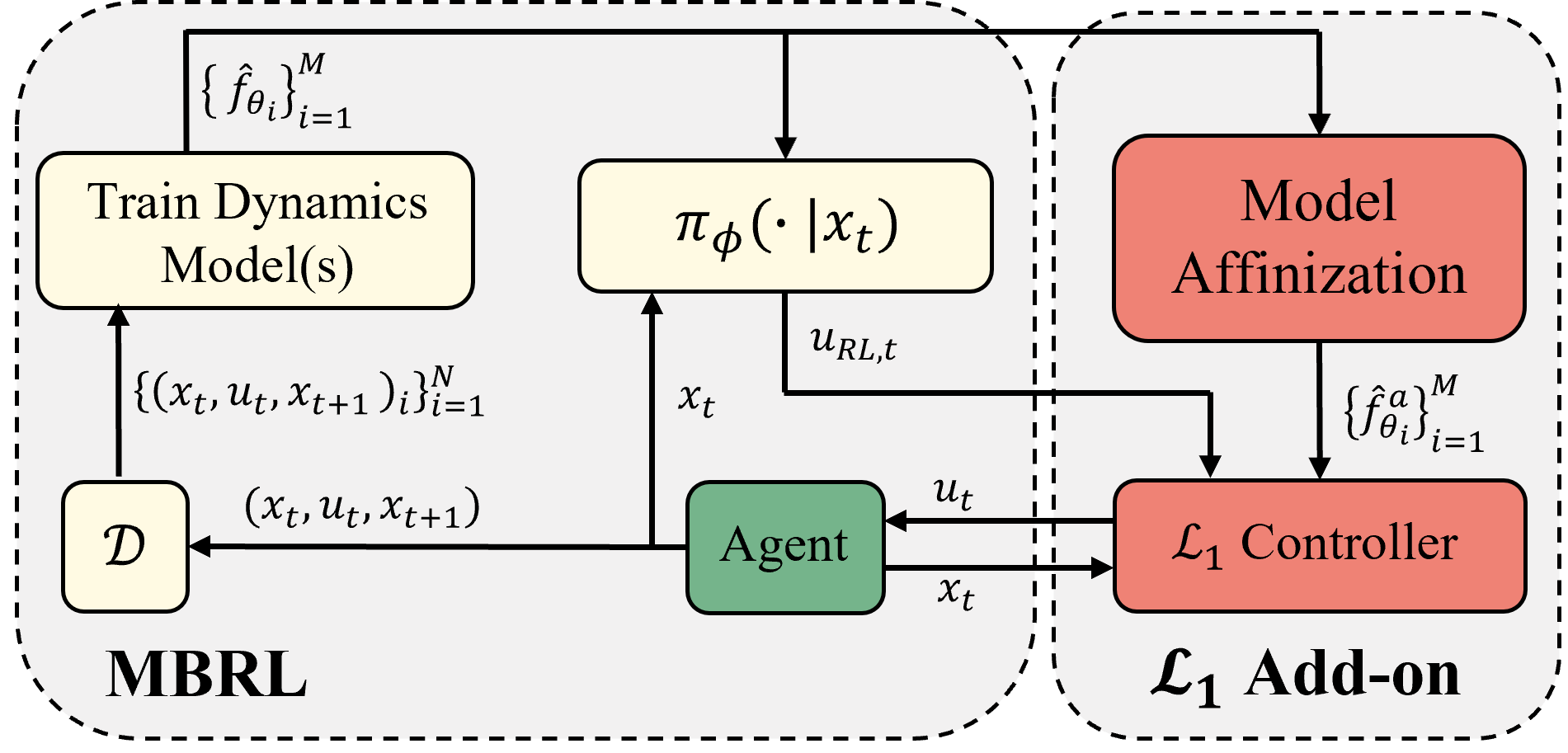}
    \vspace{-5pt}
    \caption{\lmbrl Framework. The policy box $\pi_\phi(\cdot|x_t)$ includes policy update and control input sampling for each time step. Although this figure illustrates an on-policy MBRL setting with a parameterized $\pi_\phi$ to provide a simple visualization, the framework is not limited to such class and can also be applied to off-policy algorithms or without a parameterized policy.}
    \label{fig:Framework}
    % \vspace{-1cm}
\end{center}
\vspace{-25pt}
\end{wrapfigure}
%%%%%%%%%%%%%%%%%%%%%%
\subsection{Standard MBRL}\label{MBRL}
As our work aims to develop an add-on module that enhances the robustness of an existing MBRL algorithm, we provide a high-level overview of a standard MBRL algorithm and its popular complementary techniques. 

The standard structure of MBRL algorithms involves the following steps: data collection, model updating, and policy updating using the updated model. To reduce model bias, many recent results consider an ensemble of models $\{{\hat{f}}_{\theta_i}\}_{i=1,2,\cdots, M}$, $M \in \mathbb{N}$, on the data set $\mathcal{D}$. The ensemble model $\hat{f}_{\theta_i}$ is trained by minimizing the following loss function for each $\theta_i$~\citep{nagabandi2018neural}: 
\begin{equation}\label{eq:l2_minimization}
\frac{1}{|\mathcal{D}|} \sum_{(x_t,u_t,x_{t+1}) \in D} \| (x_{t+1}-x_t) - \hat{f}_{\theta_i}(x_t, u_t)\|_2^2. 
\end{equation}
In this paper, we consider (\ref{eq:l2_minimization}) as the loss function used to train the baseline MBRL algorithm among other possibilities. This is only for convenience in explaining the \ellone augmentation in Section~\ref{subsec:l1_augmentation}, and appropriate adjustments can be readily made for the augmentation upon the choice of different loss functions.

Besides the loss function, methods like random initialization of parameters, varying model architectures, and mini-batch shuffling are widely used to reduce the correlation among the outputs of different models in the ensemble. Further, various standard techniques including early stopping, input/output normalization, and weight normalization can be used to avoid overfitting. 

Once the model is learned, control input can be computed by any of the following options: 1) using the learned dynamics as a simulator to generate fictitious samples~\citep{kurutach2018modelensemble,clavera2018model}, 2) leveraging the derivative of the model for policy search~\citep{levine2013guided,heess2015learning}, or 3) applying the Model Predictive Controller (MPC)~\citep{nagabandi2018neural,chua2018deep}. We highlight here that our proposed method is agnostic to the use of particular techniques or the choice of the policy optimizer.  
%%%%%%%%%%%%%%%%%%%
\subsection{\ellone Augmentation}\label{subsec:l1_augmentation}
Let the true dynamics in discrete time be given by
%%%%%%%%%%%%%%%%%
\begin{equation}\label{eq:gt_dynamics}
    \Delta x_{t+1} =  f(x_t,u_t)+w(t,x_t, u_t), \quad \Delta x_{t+1}  \triangleq x_{t+1}-x_t,
\end{equation}
%%%%%%%%%%%%%%%%%%%
where the transition function $f$ and the system uncertainty $w$ are unknown. Let  $\hat{f}_\theta:= \frac{1}{M} \sum_{i=1}^M f_{\theta_i}$ be the mean of the ensemble model trained using the loss function in~\Eqref{eq:l2_minimization}. Then, we express  
%%%%%%%%%%%%%%%%% 
\begin{equation}\label{eq:MBRL}
    \Delta \bar{x}_{t+1} =  \hat{f}_\theta(x_t,u_t), \quad \Delta \bar{x}_{t+1} \triangleq \bar{x}_{t+1}-x_t,
\end{equation}
%%%%%%%%%%%%%%%%%
where $\bar{x}_{t+1}$ indicates the estimate of the next state evaluated with $\hat{f}_\theta(x_t,u_t)$. In MBRL, such transition functions are typically modeled using fully nonlinear function approximators like NNs. However, as discussed in Sec.~\ref{subsec:L1-prelim}, it is necessary to represent the nominal model in the control-affine form to apply \ellone adaptive control. A common approach to obtain a control-affine model involves restricting the model structure to the control-affine class~\citep{khojasteh2020probabilistic,taylor2019episodic,choi2020reinforcement}. For NN models, this process involves training two NNs $g_\theta$ and $h_\theta$, such that~\Eqref{eq:MBRL} becomes $\Delta \bar{x}_{t+1} =  g_\theta(x_t) + h_\theta(x_t)u_t$.

While control-affine models are used for their tractability and direct connection to control-theoretic methods, they are inherently limited in their representational power compared to fully nonlinear models, and hence, their use in an MBRL algorithm can result in reduced performance.
%%%%%%%%%%%%%%%%%%%%%%
\begin{wrapfigure}{r}{0.4\textwidth}
  \vspace{-15pt}
  \begin{center}
    \includegraphics[width =0.4\textwidth]{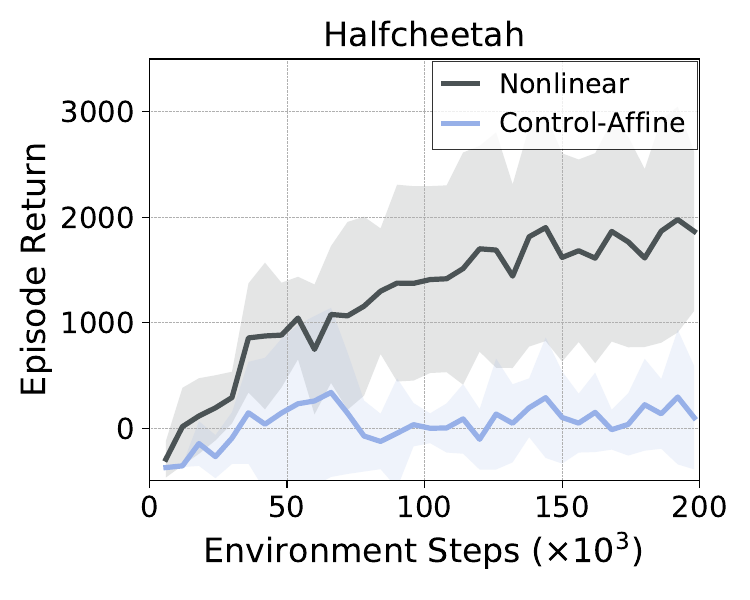}
    \vspace{-10pt}
    \caption{Comparison of performance between fully nonlinear and control-affine model on the Halfcheetah environment using METRPO.The control-affine model failed to learn the Halfcheetah dynamics.}
    \label{fig:ca_vs_nl}
  \end{center}
  \vspace{-20pt}
\end{wrapfigure}
%%%%%%%%%%%%%%%%%%%%%%
To study the level of compromise on the performance, we compare fully nonlinear models with control-affine models in the Halfcheetah environment for METRPO~ \citep{kurutach2018modelensemble}, where each size of the implicit layers of the control-affine model $g_\theta$ and $h_\theta$ are chosen to match that of the fully nonlinear $\hat{f}_\theta$. The degraded performance of the control-affine model depicted in Fig.~\ref{fig:ca_vs_nl} can be primarily attributed to intricate nonlinearities in the environment. 

Although using the above naive control-affine model can be convenient, it must trade in the capabilities of the underlying MBRL algorithm. To avoid such limitations, we adopt an alternative approach inspired by the Guided Policy Search~\citep{levine2013guided}.  Specifically, we apply a control-affine transformation to the fully nonlinear dynamics multiple times according to the predefined \emph{switching law}. Specifically, we apply the first-order Taylor series approximation around the operating input $\bar{u}$:
%%%%%%%%%%%%%%%%%
\begin{equation}\label{eq:affinzation}
\begin{aligned}
    \hat{f}_{\theta}&(x_t,u_t) \approx \hat{f}_\theta(x_t,\Bar{u}) + \left(\left[\nabla_u\hat{f}_\theta(x_t,u)\right]_{u=\Bar{u}}\right)^\top (u_t-\bar{u})\\
    &=\underbrace{\hat{f}_\theta(x_t,\Bar{u})- \left(\left[\nabla_u\hat{f}_\theta(x_t,u)\right]_{u=\Bar{u}}\right)^\top \Bar{u}}_{\triangleq~g_\theta(x_t)} + \underbrace{\left(\left[\nabla_u\hat{f}_\theta(x_t,u)\right]_{u=\Bar{u}}\right)^\top}_{\triangleq~h_\theta(x_t)} u_t \triangleq  \hat{f}^a_{\theta}(x_t,u_t;\bar{u}).
\end{aligned}    
\end{equation}
%%%%%%%%%%%%%%%%%
Here, the superscript $a$ indicates the \textit{affine} approximation of $\hat{f}_\theta$. The term \textit{affinization} in this paper is distinguished from \textit{linearization}, which linearizes the function with respect to both $x_t$ and $u_t$ such that $\bar{x}_{t+1} \simeq Ax_t + Bu_t$ for some constant matrix $A$ and $B$. Since it is common to have more states than control inputs in most controlled systems, the affinized model is a significantly more accurate approximation of the nominal dynamics compared to the linearized model. 

Indeed, the control-affine model $\hat{f}_\theta^a$ is only a good approximation of $\hat{f}_\theta$ around $\bar{u}$. When the control input deviates considerably from $\bar{u}$, the quality of the approximation deteriorates. To handle this, we produce the next approximation model when the following \emph{switching} condition holds:
\begin{equation}\label{eq:switching}
 \|\hat{f}^a_\theta(x_t,u_t;\bar{u})- \hat{f}_{\theta}(x_t,u_t)\| \geq \epsilon_a.
\end{equation}
Here, $\|\cdot\|$ indicates the vector norm, and $\epsilon_a$ is the model tolerance hyperparameter chosen by the user. Note that as $\epsilon_a \rightarrow 0$, we make an affine approximation at every point in the input space and we retrieve the original non-affine function $\hat{f}_\theta$.
\begin{remark}
    Although a more intuitive choice for the switching condition would be $\|u_t-\bar{u}\|>\epsilon_a$, we adopt an implicit switching condition (\Eqref{eq:switching}) to explicitly control over the acceptable level of prediction error between $\hat{f}^a_\theta$ and $\hat{f}_\theta$ by specifying the threshold $\epsilon_a$.  This approach prevents significant deviation in the performance of the underlying MBRL algorithm, and its utilization is instrumental in establishing the theoretical guarantees of the uncertainty estimation (See Section~\ref{sec:theoretical_analysis}).
\end{remark}

Given a locally valid control-affine model $\hat{f}_\theta^a$, we can proceed with the design of the \ellone input by utilizing the discrete-time version of the controller presented in Sec.~\ref{subsec:L1-prelim}. In particular, the state-predictor, the adaptation law, and the control law are given by
%%%%%%%%%%%%%%%%
\begin{subequations}\label{eq:lmbrl_discrete}
\begin{align}
    \xh_{t+1} &= \xh_t + \Delta \xh_t \triangleq \xh_t + \hat{f}^a_\theta(x_t,u_t) + \hat{\sigma}_t +(A_s\Tilde{x}_t)\Delta t, \quad \xh_{0} = x_0, \: \tilde{x}_{t} = \xh_t - x_t, \label{eqn:predictor}\\
    \hat{\sigma}_{t} &= -\Phi^{-1}\mu_{t},\label{eqn:adaptation_law}\\
    q_t &= q_{t-1} + (-K q_{t-1}+K \hat{\sigma}_{m,t-1})\Delta t,\quad q_0=0,~\: u_{a,t} = - q_{t},\label{eqn:control_input}   
\end{align}
\end{subequations}
%%%%%%%%%%%%%%%%
where $ K \succ 0 $ is an $m\times m$ symmetric matrix that characterizes the first order low pass filter $C(s)=K(s\mathbb{I}_m+K)^{-1}$, discretized in time. Note that~\Eqref{eqn:predictor_continuous}-\Eqref{eqn:control_law_continuous} can be considered as zero-order-hold continuous-time signals of discrete signals produced by ~\Eqref{eq:lmbrl_discrete}. As such, $\hat{\sigma}_t$ and $\mu_t$ are defined analogously to their counterparts in the continuous-time definitions. In our setting, where prior information about the desired input signal frequency is unavailable, an unbiased choice is to set $K=\omega \mathbb{I}_m$, where $\omega$ is the chosen cutoff frequency. The sampling time $T_s$ is set to $\Delta t$, which corresponds to the time interval at which the baseline MBRL algorithm operates. The algorithm for the \ellone adaptive control is presented in Algorithm~\ref{algo:L1}.
% Lastly, the state predictor in~\Eqref{eqn:predictor} predicts the state at the next time step $t+1$, which is then used in~\Eqref{eqn:control_input} to determine the next control input augmentation $u_{a,t+1}$. 

%%%%%%%%%%%%%%%%%%%%%%%%%%%%%%%%%%%%
\begin{wrapfigure}{R}{0.5\textwidth}
\begin{minipage}{0.5\textwidth}
\vspace{-13pt}
\begin{algorithm}[H]
    \caption{\sc  $\mathcal{L}_1$ Adaptive Control}
    \small
    \SetAlgoLined
    \DontPrintSemicolon
    \KwData{Initialize $\{\hat{x}_t,~\hat{\sigma}_{t}\} \leftarrow 0 $ \\
    Set $\omega$ for $K$ in~\Eqref{eqn:control_input} \\
    % Obtain affine $\hat{f}_\theta^{a}$ from  $\hat{f}_\theta$~\Eqref{eq:affinzation} 
    }
        \SetKwFunction{FMain}{Control}
        \SetKwProg{Fn}{Function}{:}{}
        \Fn{\FMain{$u_{RL,t},x_t,\hat{f}^a_\theta$ }}{
            \BlankLine
            \vspace{-4pt}
            Prediction error update: $\Tilde{x}_t \longleftarrow  \xh_t - x_t $
            \BlankLine
            \vspace{-4pt}
            Uncertainty estimate $\hat{\sigma}_{t}$ update:~\Eqref{eqn:adaptation_law}
            \BlankLine
            \vspace{-4pt}
            Compute $ u_{a,t}~($\Eqref{eqn:control_input})
            \BlankLine
            \vspace{-4pt}
            $ u_{t} \leftarrow  u_{RL,t} + u_{a,t} $
            \BlankLine
            \vspace{-4pt}
            Update $\xh_{t+1}$~(\Eqref{eqn:predictor})
            \BlankLine
            \vspace{-4pt}
            \textbf{return} $ u_{t} $
            \vspace{4pt}
    }
    \textbf{End Function}
    \label{algo:L1}
\end{algorithm}
\vspace{-15pt}
\end{minipage}   
\end{wrapfigure}
%%%%%%%%%%%%%%%%%%%%%%%%%%%%%%

As the underlying MBRL algorithm updates its model $\hat{f}_\theta$, the corresponding control-affine model $\hat{f}^a_\theta$ and \ellone control input $u_a$ are updated sequentially (Algorithm~\ref{algo:L1}). 
By incorporating the \ellone control augmentation into the MBRL algorithm, we obtain the \lmbrl algorithm, as outlined in Algorithm~\ref{algo:L1-MBRL}. Note that in this algorithm we are adding $u_{RL,t}$ instead of $u_t$ to the dataset. Intuitively, this is to learn the nominal dynamics that remains after uncertainties get compensated by $u_{a,t}$. Similar approach has been employed previously in~\cite[Appendix A.1]{Wang2020Exploring}. 

%%%%%%%%%%%%%%%%%%%%%%%%%%%%
%%%%%L1-MBRL ALGORITHM%%%%%%
%%%%%%%%%%%%%%%%%%%%%%%%%%%%%%%%%%%%%%
\begin{wrapfigure}{R}{0.5\textwidth}
\begin{minipage}{0.5\textwidth}
\vspace{-14pt}
\begin{algorithm}[H]
\caption{{\sc \lmbrl Algorithm}}
\small
\DontPrintSemicolon % Some LaTeX compilers require you to use \dontprintsemicolon instead
Set $\mathcal{D}\leftarrow {\emptyset}$, $\{\hat{x}_t,~\hat{\sigma}_{t}\} \leftarrow 0$\\
Initialize $\epsilon,\pi_\phi, \hat{f}_\theta, \omega,x_0$  \\
\Repeat{the average return converges\;}{
\For{Episodes $N_e \in \mathbb{N}$}{
$\bar{u} \leftarrow \texttt{None} $ \\
\For{Horizon $H \in \mathbb{N}$}{
$u_{RL,t} \sim \pi_\phi(\cdot|x_t)$ \\
\If{$\bar{u}$ is $\texttt{None}$ or~\Eqref{eq:switching}}{
\vspace{4pt}
$\bar{u}\leftarrow u_{RL,t}$\\
compute $\hat{f}_\theta^{a}$ via~\Eqref{eq:affinzation}}
% \If{$\|\hat{f}_\theta(x_t,u_{RL,t})-  \hat{f}_\theta^{a}(x_t,u_{RL,t},\bar{u})\| > \epsilon$ }{
% \vspace{4pt}
% $\bar{u} = u_{RL,t}$, compute $\hat{f}_\theta^{a}$~\Eqref{eq:affinzation}}
$u_{t} \leftarrow \texttt{Control}$ in Algo.~\ref{algo:L1} \\
Execute $u_{t}$ and  $\mathcal{D} \leftarrow (x_t,u_{RL,t},x_{t+1})$ }}
\vspace{-4pt}
Update model(s) $\hat{f}_\theta$ using $\mathcal{D}$\\
Update policy $\pi_\phi$
}
% \vspace{-10pt}
\label{algo:L1-MBRL}
\end{algorithm}
\vspace{-10pt}
\end{minipage}   
\end{wrapfigure}
%%%%%%%%%%%%%%%%%%%%%%%%%%%%

Our \lmbrl framework makes a control-affine approximation of the learned dynamics, which is itself an approximation of the ground truth dynamics. Such layers of approximations may amplify errors, which may degrade the effect of the \ellone augmentation. In this section, we prove that the \ellone prediction error is bounded, and subsequently, the \ellone controller effectively compensates for uncertainties. To this end, we conduct a continuous-time analysis of the system that is controlled via \lmbrl framework which operates in sampled-time~\citep{aastrom2013computer}. It is important to note that our adaptation law (\Eqref{eqn:adaptive_law_continuous}) holds the estimation for each time interval (zero-order-hold), converting discrete estimates obtained from the MBRL algorithm into a continuous-time signal. Such choice of the adaptation law ensures that the \ellone augmentation is compatible with the discrete MBRL setup, providing the basis for the following analysis.

\subsection{Theoretical analysis}\label{sec:theoretical_analysis}

Consider the nonlinear (unknown) continuous-time counterpart of~\Eqref{eq:gt_dynamics}
\begin{equation}\label{eq:ground_truth_nonlinear}
    \dot{x}(t) = F(x(t),u(t)) + W(t,x(t), u(t))\footnote{The continuous-time functions correspond to the Euler-integral of its discrete counterparts.},
\end{equation}
where $F:\mathcal{X}\times\mathcal{U}\rightarrow \mathbb{R}^n$ is a fully nonlinear function defining the vector field. Note that unlike the system in~\Eqref{eqn:true_continuous}, we do not make any assumptions on $F(x(t),u(t))$ being control-affine. Furthermore, recall from Sec.~\ref{subsec:MBRL_prelim} that the sets $\mathcal{X} \subset \mathbb{R}^n$ and $\mathcal{U} \subset \mathbb{R}^m$, over which the MBRL experiments take place, are compact. Additionally, $W(t,x(t), u(t))\in \mathbb{R}^n$ represents the disturbance perturbing the system.  As before, we denote by $\hat{F}_\theta(x(t),u(t))$ the approximation of $F(x(t),u(t))$, and its affine approximate as 
\begin{equation}\label{eq:continuous_control_affine}
    \hat{F}^{a}_\theta(x(t),u(t)) ={G_\theta(x(t))+H_\theta(x(t))u(t)}.
\end{equation} Subsequently, we define the residual error $l(t,x(t),u(t))$ and an affinization error $a(x(t),u(t))$ as 
\begin{align*}
    l(t,x(t),u(t)) &\triangleq F(x(t),u(t))+W(t,x(t), u(t))-\hat{F}_\theta(x(t),u(t))\\
    a(x(t),u(t)) &\triangleq \hat{F}_\theta(x(t),u(t)) - \hat{F}^{a}_\theta(x(t),u(t)).
\end{align*}
Note that $\|a(x(t),u(t))\|\leq \epsilon_a$ in the \lmbrl framework by~\Eqref{eq:switching}. 

We pose the following assumptions.
\begin{assumption}\label{assumption}\hfill
    \begin{enumerate}
        \item The functions $F(x(t),u(t))$ and $W(t,x(t), u(t))$ are Lipschitz continuous over $\mathcal{X} \times \mathcal{U}$ and $[0,t_{\max}) \times \mathcal{X} \times \mathcal{U}$, respectively, for $0 < t_{\max} \leq H$, where $H$ is a known finite time horizon of the episode. The learned model $\hat{F}_\theta(x(t),u(t))$ is Lipschitz continuous in $\mathcal{X}$, and continuously differentiable\footnote{More precisely, $\mathcal{C}^1$ everywhere except finite sets of measure zero.} ($\mathcal{C}^1$) in $\mathcal{U}$.
        
        \item The learned model is uniformly bounded over $(t,x,u) \in [0,t_{\max}) \times \mathcal{X} \times \mathcal{U}$: 
        \begin{align}
            \|F(x(t),u(t))+W(t,x(t), u(t))-\hat{F}_\theta(x(t),u(t))\|&\leq \epsilon_l,\label{eq:nn_model_error}
        \end{align}
        where the bound $\epsilon_l$ is assumed to be known. 
        % \item 
    \end{enumerate}
    
\end{assumption}

See Appendix~\ref{appendix:assumption} for remarks on this assumption.

Next, we set
\begin{equation}\label{eq:final_augmented_input}
u(t) = \us(t) + u_a(t),
\end{equation}
where $\us(t)$ is the continuous-time signal converted from the discrete control input produced by the underlying MBRL, and $u_a(t)$ is the \ellone control input. As described in Section~\ref{subsec:L1-prelim}, the \ellone controller estimates the uncertainty by following the piecewise constant adaptive law~(\Eqref{eqn:adaptive_law_continuous}). Now, we aim to evaluate the estimation error $e(t,x(t),u(t))$:
\begin{align}
    e(t,&x(t),u(t)) \triangleq l(t,x(t),u(t)) +  a(x(t),u(t)) - \hat{\sigma}(t) \notag \\
   = &\label{eq:final_error} F(x(t),u(t)) + W(t,x(t), u(t)) - G_\theta(x(t)) - H_\theta(x(t))u(t) - \hat{\sigma}(t).
\end{align}
 Our interest in evaluating the estimation error is articulated in Remark~\ref{rem:Thm1}, Appendix~\ref{appendix:theoretical_analysis}, where we also provide proof of the following theorem. We note here that the sets $\mathcal{X}$ and $\mathcal{U}$ in the following result are compact due to the inherent nature of the MBRL algorithm, as described in Sec.~\ref{subsec:MBRL_prelim}.
%%%%%%%%%%%%%%%%%%%%%%%%%%%%%%%%%%%%%%%%%%%%%%%%%%%%
\begin{theorem}\label{theorem}
Consider the system described by~\Eqref{eq:ground_truth_nonlinear}, and its learned control-affine representation in~\Eqref{eq:continuous_control_affine}. Additionally, assume that the system is operating under the augmented feedback control presented in~\Eqref{eq:final_augmented_input}. Let  $A_s = \texttt{diag}\{\lambda_1, \dots, \lambda_n\}$ be the Hurwitz matrix that is used in the definition of the state predictor~(\Eqref{eqn:predictor_continuous}). If Assumption~\ref{assumption} holds, then the estimation error defined in~\Eqref{eq:final_error} satisfies $\|e(t,x(t),u(t)\| \leq \epsilon_l+\epsilon_a$, $ \forall t\in [0,T_s)$ and
\begin{align*}
    \|e(t,x(t),u(t))\| = 2\epsilon_a + \mathcal{O}(T_s)\quad \forall t \in [T_s,t_{\max}), 
\end{align*}
 where $0< T_s < t_{\max}  \leq H < \infty$, and $H$ is the known bounded horizon (see Sec.~\ref{subsec:MBRL_prelim}).
    
    % and $\lambda_{\min}$ is the eigenvalue of $A_s$ that has the minimium absolute value. 

\end{theorem}
% The proof of the theorem and its interpretation are provided in Appendix~\ref{appendix:theoretical_analysis}.

%%%%%%%%%%%%%%%%%%%%%%%%%%%%
\section{Simulation Results}
In this section, we demonstrate the efficacy of our proposed $\mathcal{L}_{1}$-MBRL framework using the METRPO~\citep{kurutach2018modelensemble} and MBMF~\citep{nagabandi2018neural} as the baseline MBRL method~\footnote{METRPO and MBMF conform to the standard MBRL framework but employ different strategies for control optimization (refer to Section~\ref{subsec:MBRL_prelim}). Selecting these baselines demonstrates that our framework is agnostic to various control optimization methods, illustrating its functionality as a versatile add-on module.}, and we defer the \ellone-MBMF derivations and details to  Appendix~\ref{sec:extended_simulation}. We briefly note here that we observed a similar performance improvement with \ellone augmentation as for \ellone-METRPO.

% \begin{remark} 
%     Different MBRL algorithms have different motivations and objectives, so it is common to see algorithms perform well in some tasks while poorly performing in others. Our simulation is not intended to demonstrate state-of-the-art performance in the chosen scenarios, but to compare the performance of the baseline MBRL with its \ellone augmented counterpart. For this reason, 
% \end{remark}

In our first experimental study, we evaluate the proposed $\mathcal{L}_{1}$-MBRL framework on five different OpenAI Gym environments~\citep{brockman2016openai} with varying levels of state and action complexity. For each environment, we report the mean and standard deviation of the average reward per episode across multiple random seeds. Additionally, we incorporate noise into the observation ($\sigma_o = 0.1$) or action ($\sigma_a = 0.1$) by sampling from a uniform distribution~\citep{wang2019benchmarking}. This enables us to evaluate the impact of noise on MBRL performance. The results are summarized in Table~\ref{tab:main_result}. Further details of the experiment setup are provided in the Appendix~\ref{sec:extended_simulation}.

\vspace{-5pt}
\begin{table*}[h]
{\renewcommand{\arraystretch}{1.3}
  \footnotesize
  \centering
  \caption{Performance comparison between METRPO and $\mathcal{L}_1$-METRPO (Ours). The average performance and standard deviation over multiple seeds are evaluated for a window size of 3000 timesteps at the end of the training for multiple seeds. Higher performance cases are marked in bold and green.}
  \label{tab:main_result}
  \vspace*{-0.5em}
  \addtolength{\tabcolsep}{-2pt}
  \begin{center}
  \resizebox{\columnwidth}{!}{%
  \begin{tabular}{c|cc|cc|cc}
    \toprule
	\textbf{ } & 
	\multicolumn{2}{c|}{\bfseries \footnotesize Noise-free} & 
	\multicolumn{2}{c|}{\bfseries \footnotesize $\mathbf{\sigma_a=0.1}$} &
	\multicolumn{2}{c}{\bfseries \footnotesize $\mathbf{\sigma_o=0.1}$}\\ 
	\textbf{\footnotesize Env.} 
	
	& {\footnotesize METRPO} & {\footnotesize $\mathcal{L}_1$-METRPO} & {\footnotesize METRPO} & {\footnotesize $\mathcal{L}_1$-METRPO} & {\footnotesize METRPO} & {\footnotesize $\mathcal{L}_1$-METRPO}
	\\ \toprule

	{\footnotesize Inv. P.}
    & $-51.3 \pm 67.8$     & \textcolor{applegreen}{$\mathbf{-0.0 \pm 0.0}$}     & $-105.2\pm 81.6$     & \textcolor{applegreen}{$\mathbf{-0.0 \pm 0.0}$}     & $-74.22 \pm 74.5$     & \textcolor{applegreen}{$\mathbf{-21.3 \pm 20.7}$}\\ \midrule

 	{\footnotesize Swimmer}
 	&$309.5 \pm 49.3$     &\textcolor{applegreen}{$\mathbf{313.8 \pm 18.7}$}       & $258.7\pm 113.7$       & \textcolor{applegreen}{$\mathbf{322.7 \pm 5.3}$}        & $30.7 \pm 56.1$     & \textcolor{applegreen}{$\mathbf{79.2 \pm 85.0}$}\\ \midrule
 	
    {\footnotesize Hopper}
 	& $1140.1 \pm 552.4$     & \textcolor{applegreen}{$\mathbf{1491.4 \pm 623.8}$}     & $609.0\pm793.5$   & \textcolor{applegreen}{$\mathbf{868.7\pm735.8}$}     & $-1391.2\pm 266.5$     & \textcolor{applegreen}{$\mathbf{-486.6 \pm 459.9}$} \\ \midrule
 	
    {\footnotesize Walker}
 	& \textcolor{applegreen}{$\mathbf{-6.6 \pm 0.3}$}      & $-6.9 \pm 0.5$   & $-9.8 \pm 2.2$     & \textcolor{applegreen}{$\mathbf{-5.9 \pm 0.3}$}     & $-30.3 \pm 28.2$     & \textcolor{applegreen}{$\mathbf{-6.3 \pm 0.3}$}\\ \midrule
 	
    {\footnotesize Halfcheetah}
 	& $2367.3\pm 1274.5$     & \textcolor{applegreen}{$\mathbf{2588.6\pm 955.1}$}    & $1920.3 \pm 932.4$     & \textcolor{applegreen}{$\mathbf{2515.9 \pm 1216.4}$}     & $1419.0\pm 517.2$     & \textcolor{applegreen}{$\mathbf{1906.3\pm 972.7}$}\\

    \bottomrule
	\end{tabular}}
	\end{center}}
    \vspace*{-0.5em}
    % \label{table:results}
\end{table*}

%%%%%%%%%%%%%%%%%
The experimental results demonstrate that our proposed \lmbrl framework outperforms the baseline METRPO in almost every scenario. Notably, the advantages of the \ellone augmentation become more apparent under noisy conditions.

\subsection{Ablation Study}
We conduct an ablation study to compare the specific contributions of \ellone control in the training and testing. During testing, \ellone control explicitly rejects system uncertainties and improves performance. On the other hand, during training, \ellone additionally influences the learning by shifting the training dataset distribution. To evaluate the effect of \ellone at each phase, we compare four scenarios: 1) no \ellone during training or testing, 2) \ellone applied only during testing, 3) \ellone used only during training, and 4) \ellone applied during both training and testing. The results are summarized in  Fig.~\ref{fig:ablation_barplot}. 
\begin{figure}[!htb]
    \centering
    \includegraphics[width = 0.9\textwidth]{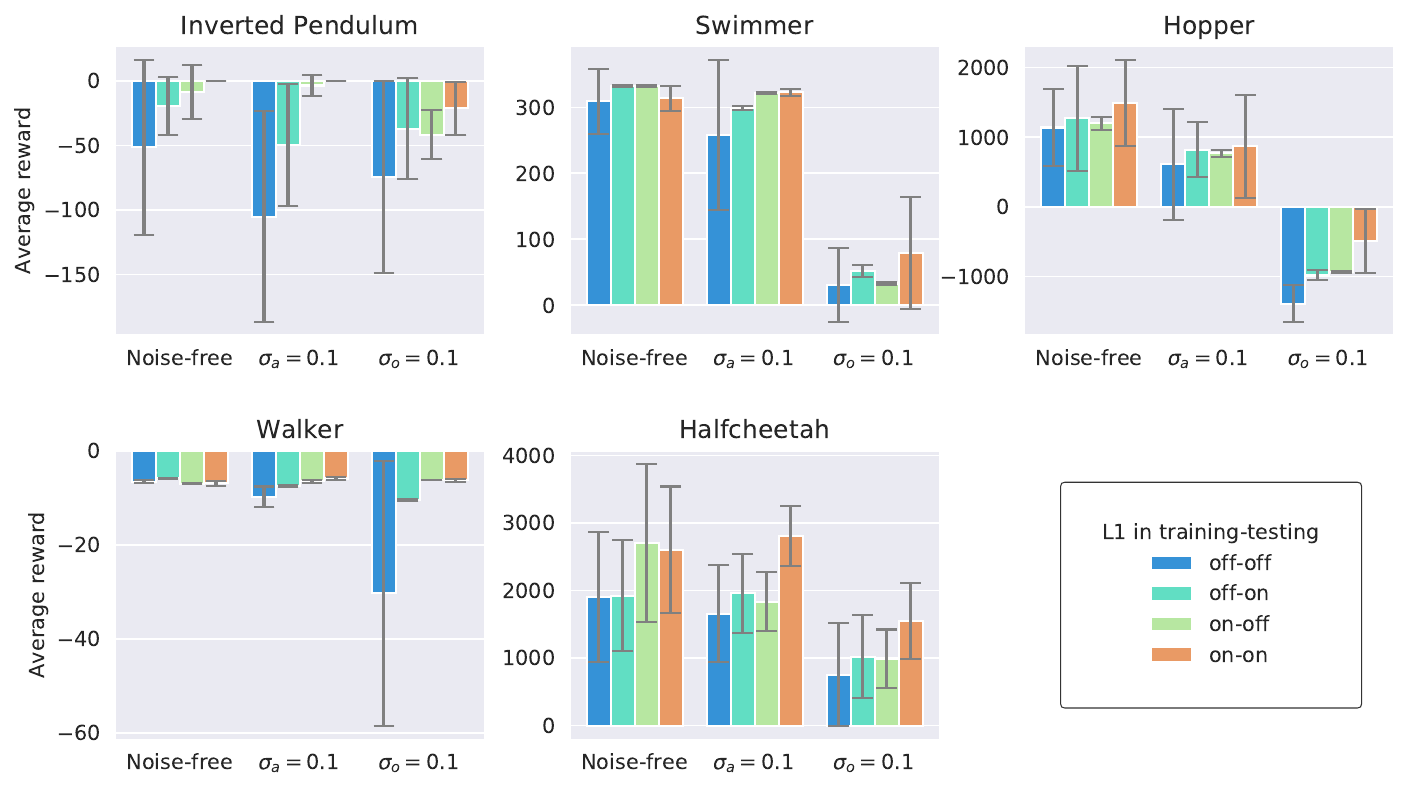}
    \caption{Contribution of \ellone in the training and testing phase. The notation \ellone on (off)-on (off) indicates \ellone is applied (not applied) during training-testing, respectively. The error bar ranges for one standard deviation of the performance. On-on and off-off correspond to our main result in Table~\ref{tab:main_result}. As expected, the on-on case achieved the highest performance in most scenarios.}
    \label{fig:ablation_barplot}
    \vspace{-5pt}
\end{figure}

The results depicted in Fig.~\ref{fig:ablation_barplot} demonstrate that the influence of \ellone control during training and testing varies across different environments and noise types. However, as anticipated, the highest performance is achieved when \ellone is applied in both the training and testing phases.

In order to evaluate the effectiveness of the \ellone-MBRL framework in addressing the sim-to-real gap, we conducted a secondary ablation study. First, we trained the model without \ellone in a noise-free environment and subsequently tested the model with and without \ellone under a noisy environment. The results are demonstrated in Table~\ref{table:sim2real}. This result indicates that our \lmbrl framework effectively addresses the sim-to-real gap, and this demonstrates the potential for directly extending our framework to the offline MBRL setting, presenting promising opportunities for future research.

%%%%%%%%%%%%%%%%%
\begin{table}[h]
  \footnotesize
  \centering
  \caption{Addressing the sim-to-real gap with \ellone augmentation: The original METRPO was initially trained on an environment without uncertainty. Subsequently, the policy was deployed in a noisy environment that emulates real-world conditions, with and without \ellone augmentation.}
  \label{tab:Reward Table}
  \addtolength{\tabcolsep}{-2pt}
  \begin{center}
  \resizebox{0.99\columnwidth}{!}{%
  \begin{tabular}{c|cc|cc|cc}

    \toprule
	\textbf{ } & 
	\multicolumn{2}{c|}{\bfseries \footnotesize $\mathbf{\sigma_a=0.1}$} &
	\multicolumn{2}{c}{\bfseries \footnotesize $\mathbf{\sigma_o=0.1}$} & 
    \multicolumn{2}{c}{\bfseries \footnotesize $\mathbf{\sigma_a=0.1}$~\&~$\mathbf{\sigma_o=0.1}$ }\\ 
	\textbf{\scriptsize Env.} 
	
	& {\footnotesize METRPO} & {\footnotesize $\mathcal{L}_1$-METRPO} & {\footnotesize METRPO} & {\footnotesize $\mathcal{L}_1$-METRPO} & {\footnotesize METRPO} & {\footnotesize $\mathcal{L}_1$-METRPO}\\ \toprule

	{\footnotesize Inv. P.}
    & $30.2\pm 45.1$     & \textcolor{applegreen}{$\mathbf{-0.0 \pm 0.0}$}     & $-74.1 \pm 53.1$     & \textcolor{applegreen}{$\mathbf{-3.1 \pm 2.0}$}    &  $-107.0 \pm 72.4$     & \textcolor{applegreen}{$\mathbf{-6.1 \pm 4.6}$}\\ \midrule

 	{\footnotesize Swimmer}
  & $250.8\pm 130.2$       & \textcolor{applegreen}{$\mathbf{330.5 \pm 5.7}$}
       & \textcolor{applegreen}{$\mathbf{337.8 \pm 2.9}$}    & $331.2 \pm 8.34$ &  $248.2 \pm 133.6$     & \textcolor{applegreen}{$\mathbf{327.3 \pm 6.8}$}\\ \midrule
 	
    {\footnotesize Hopper}
 	& $198.9\pm 617.8$   & \textcolor{applegreen}{$\mathbf{623.4\pm 405.6}$}     & $-84.5\pm 1035.8$     & \textcolor{applegreen}{$\mathbf{ 157.1 \pm 379.7}$} &  $87.5 \pm 510.2$     & \textcolor{applegreen}{$\mathbf{309.8 \pm 477.8}$}\\ \midrule
 	
    {\footnotesize Walker}
 	& \textcolor{applegreen}{$\mathbf{-6.0 \pm 0.8}$} & $-6.3 \pm 0.7$      & $-6.4 \pm 0.4$     & \textcolor{applegreen}{$\mathbf{-6.08 \pm 0.6}$} &  $-6.3 \pm 0.4$     & \textcolor{applegreen}{$\mathbf{-5.2 \pm 1.5}$}\\ \midrule
 	
    {\footnotesize Halfcheetah}
 	 & $1845.8 \pm 600.9$     & \textcolor{applegreen}{$\mathbf{1965.3 \pm 839.5}$}     & $1265.0\pm 440.8$     & \textcolor{applegreen}{$\mathbf{1861.6\pm 605.5}$}& $1355.0\pm 335.6$     & \textcolor{applegreen}{$\mathbf{1643.6\pm 712.5}$}\\
    \bottomrule
 	\end{tabular}}
	\end{center}
    \vspace*{-12pt}
    \label{table:sim2real}
\end{table}
%%%%%%%%%%%%%%%%%
\section{Limitations}
{\bf (Performance of the base MBRL)} Our \lmbrl scheme rejects uncertainty estimates derived from the \emph{learned} nominal dynamics. As a result, the performance of \lmbrl is inherently tied to the baseline MBRL algorithm, and \ellone augmentation cannot independently guarantee good performance. This can be related to the role of $\epsilon_l$ in~\Eqref{eq:nn_model_error}. Empirical evidence in Fig.~\ref{fig:ablation_barplot} illustrates this point that, despite meaningful improvements, the performance of scenarios augmented with \ellone is closely correlated to that of METRPO without \ellone augmentation.

% {\bf (Computation)} Due to multiple model affinizations, the \ellone augmentation law needs to be updated multiple times, which involves inverting a matrix in~\Eqref{eqn:adaptation_law}. Computing a matrix inverse is well-known to have a time complexity of $O(n^3)$. Therefore, additional computational resources might be required if a high-dimensional system requires frequent affinizations.  

{\bf (Trade-off in choosing $\epsilon_a$)} In Sec.~\ref{subsec:l1_augmentation}, we mentioned that as $\epsilon_a$ approaches zero, the baseline MBRL is recovered. This implies that for small values of $\epsilon_a$, the robustness properties exhibited by the \ellone control are compromised. Conversely, if $\epsilon_a$ is increased excessively, it permits significant deviations between the control-affine and nonlinear models, potentially allowing for larger errors in the state predictor (see Section~\ref{sec:theoretical_analysis}). Our heuristic observation from the experiments is to select an $\epsilon_a$ that results in approximately 0-100 affinization switches per 1000 time steps for systems with low complexity ($n<5$) and 200-500 switches for more complex systems.

% {\bf(Unmatched uncertainty)} Our application of the \ellone adaptive controller in the context of \lmbrl is limited in its ability to address uncertainties outside the range of the input operator, potentially reducing the effectiveness of \ellone compensation in the presence of large unmatched uncertainties. For linear~\cite{hovakimyan2010L1,li20121}, linear time-varying~\cite{zhao2020robust}, and certain nonlinear reference models~\cite{ackerman2021l1}, the \ellone control can minimize the effect of the unmatched uncertainties on a desired performance output. However, in more general case as presented in this paper, the compensation of the unmatched uncertainties will have to be established which is the subject of future work. 

\section{Conclusion}
In this paper, we proposed an \ellone-MBRL control theoretic add-on scheme to robustify MBRL algorithms against model and environment uncertainties. We affinize the trained nonlinear model according to a switching rule along the input trajectory, enabling the use of \ellone adaptive control. Without perturbing the underlying MBRL algorithm, we were able to improve the overall performance in almost all scenarios with and without aleatoric uncertainties. 

The results open up interesting research directions where we wish to test the applicability of \lmbrl on offline MBRL algorithms to address the sim-to-real gap~\citep{kidambi2020morel,yu2020mopo}. Moreover, its consolidation with a distributionally robust optimization problem to address the distribution shift is of interest. Finally, we will also research the \lmbrl design for MBRL algorithms with probabilistic  models~\citep{chua2018deep, Wang2020Exploring} to explore a method to utilize the covariance information in addition to mean dynamics.

\subsubsection*{Acknowledgments}
This work is financially supported by National Aeronautics and Space Administration (NASA) ULI  (80NSSC22M0070), NASA USRC (NNH21ZEA001N-USRC), Air Force Office of Scientific Research (FA9550-21-1-0411), National Science Foundation (NSF) AoF Robust Intelligence (2133656), NSF CMMI (2135925), and NSF SLES
(2331878).

\bibliography{iclr2024_conference}
\bibliographystyle{iclr2024_conference}
\newpage
\appendix
\section*{Appendix}
\section{Extended Description of \ellone Adaptive Control}\label{sec:ellone-extended}

In this section, we provide a detailed explanation of \ellone adaptive control.
Consider the following system dynamics:
\begin{align}\label{eq:true_dynamics}
    \Dot{x}(t) = g(x(t))+ h(x(t))u(t) +d(t,x(t), u(t)),~\quad x(0) = x_0,
\end{align}
where $x(t)\in \mathbb{R}^n$ is the system state vector, $u(t)\in \mathbb{R}^m$ is the control signal, $g:\mathbb{R}^n\rightarrow \mathbb{R}^n$ and $h:\mathbb{R}^n\rightarrow \mathbb{R}^{n\times m}$ are known functions that define the desired dynamics, both of which are locally-Lipschitz continuous functions. Furthermore, $d(t,x(t), u(t))\in \mathbb{R}^n$ represents the unknown nonlinearities and is continuous in its arguments. 
We now decompose $d(t,x(t), u(t))$ with respect to the range and kernel of $h(x(t))$ to obtain
\begin{align}\label{eq:true_dynamics_matched}
    \Dot{x}(t) = g(x(t))+ h(x(t))(u(t)+\sigma_m(t,x(t), u(t))) +h^\perp(x(t))\sigma_{um}(t,x(t), u(t)), \quad x(0)=x_0,
\end{align}
where $h(x(t))\sigma_m(t,x(t), u(t)) +h^\perp(x(t))\sigma_{um}(t,x(t), u(t))=d(t,x(t), u(t))$. Moreover, $h^\perp (x(t))\in \mathbb{R}^{n\times (n-m)}$ denotes a matrix whose columns are perpendicular to $h(x(t))\in \mathbb{R}^{n\times m}$, such that $ h(x(t))^\top h^\perp(x(t)) = 0$ for any $x(t)\in \mathbb{R}^n$. The existence of $h^\perp(x(t))$ is guaranteed if it is a full-rank matrix. The terms $\sm(t,x(t))$ and $\summ(t,x(t))$ are commonly referred to as \emph{matched} and \emph{unmatched} uncertainties, respectively. 

Consider the \emph{nominal} system in the absence of uncertainties
\begin{align*}
     \Dot{x}^{\star}(t) = g(\xs(t))+ h(\xs(t))\us(t),~\quad \xs(0) = x_0,
\end{align*}
%%%%%%%%%%%%%%%%%%%%%%%%%
\begin{wrapfigure}{r}{0.55\textwidth}
% \vspace{-10pt}
\begin{center}
    \centering
    \includegraphics[width=0.55\textwidth]{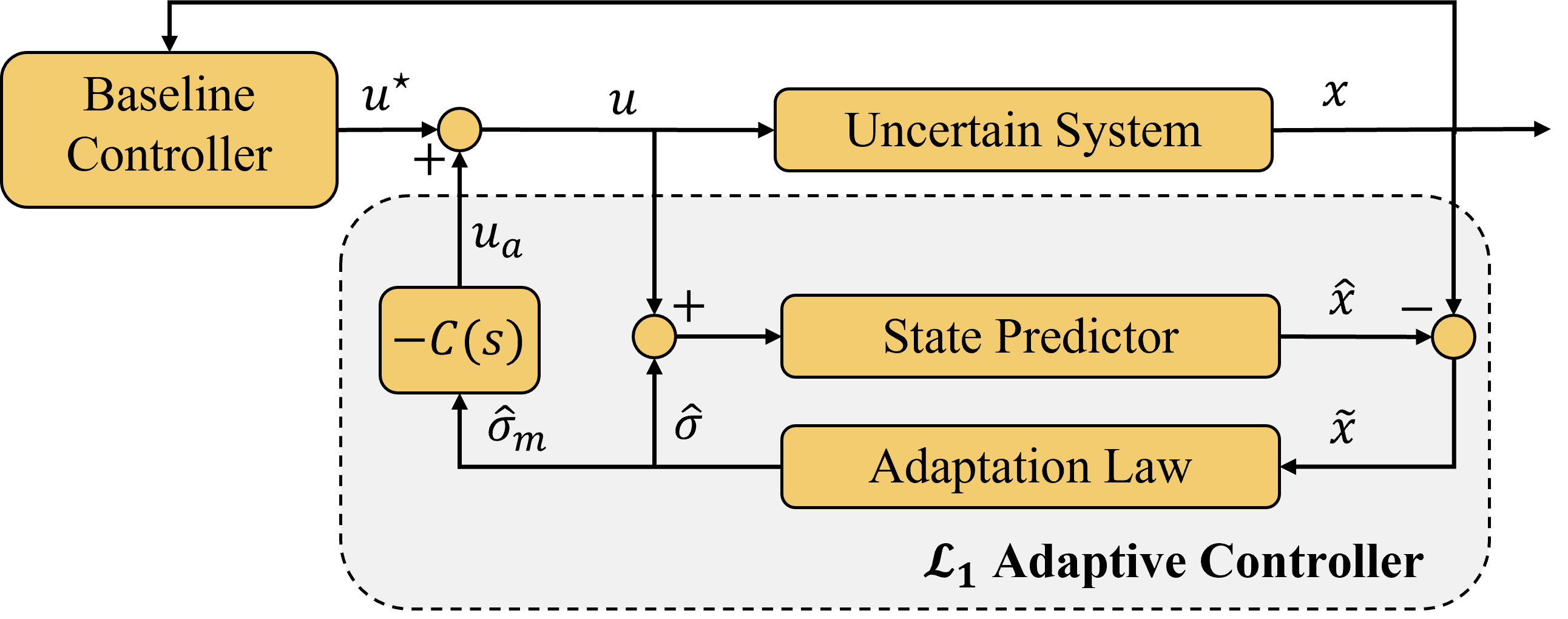}
    \caption{The architecture of  \ellone adaptive controller.}
    \label{fig:l1_architecture}
\end{center}
\vspace{-5pt}
\end{wrapfigure}
%%%%%%%%%%%%%%%%%%%%%
where $\us(t)$ is the baseline input designed so that the desired performance and safety requirements are satisfied. If we pass the baseline input to the true system in~\Eqref{eq:true_dynamics}, the actual state trajectory $x(t)$ can diverge from the desired state trajectory $\xs(t)$ in an unquantifiable manner due to the presence of uncertainties $d(t,x(t), u(t))$. To avoid this behavior, we employ \ellone adaptive control, which aims to compute an input $u_a(t)$ such that, when combined with the nominal input $\us(t)$, forms the \emph{augmented input}
\begin{equation}\label{eq:augmented_input}
    u(t) = \us(t) + u_a(t).
\end{equation} The objective of this approach is to ensure that the true state $x(t)$ remains uniformly and quantifiably bounded around the nominal trajectory $\xs(t)$ under certain conditions.

Next, we explain how the \ellone adaptive controller achieves this goal. The \ellone adaptive control has three components: the state predictor, the adaptation law, and a low-pass filter. The \emph{state predictor} for~\Eqref{eq:true_dynamics} is given by
\begin{equation}\label{eq:state_predictor}
    \dot{\xh}(t) = g(x(t))+h(x(t))u(t) + \hat{\sigma}(t)+ A_s \Tilde{x}(t),
\end{equation}
where $\hat{\sigma}(t) \triangleq h(x(t))\hat{\sigma}_m(t) + h^\perp(x(t))\hat{\sigma}_{um}(t)$. Moreover, $\smh(t)$ and $\summh(t)$ are the estimates of the matched and unmatched uncertainties $\sm(t,x(t))$ and $\summ(t,x(t))$, respectively. The initial conditions are given by $\hat{x}(0)= x_0$,  $\smh(0)=0$, and $\summh(0)=0$. Here, $\xh(t)\in \mathbb{R}^n$ is the state of the predictor, $u(t) \in \mathbb{R}^m$ is the augmented control input in~\Eqref{eq:augmented_input}, $\Tilde{x}(t) = \xh(t) - x(t)$ denotes the \emph{prediction error}, and $A_s \in \mathbb{S}^n$ is a chosen Hurwitz matrix. The state predictor produces the state estimate $\xh(t)$ induced by the adaptive estimates $\hat{\sigma}_m(t)$ and $\hat{\sigma}_{um}(t)$. 

Following the true dynamics in~\Eqref{eq:true_dynamics} and the state predictor in~\Eqref{eq:state_predictor}, the dynamics of the prediction error $\tilde{x}(t) = \xh(t) - x(t)$ is given by 
\begin{equation}\label{eq:prediction_error_dynamics}  
\dot{\Tilde{x}}(t) 
    = A_s \tilde{x}(t) + \left[\hat{\sigma}(t) - d(t,x(t), u(t)) \right], \quad \tilde{x}(0) = 0.
\end{equation}
Here, we refer to $\hat{\sigma}(t) - d(t,x(t), u(t))$ as the \emph{uncertainty estimation error}. Moreover, since $A_s$ is Hurwitz, the system~\Eqref{eq:prediction_error_dynamics} governing the evolution of $\tilde{x}(t)$ is exponentially stable in the absence of the exogenous input $\hat{\sigma}(t) - d(t,x(t), u(t))$. Therefore, $\tilde{x}(t)$ serves as a \emph{learning signal} for the adaptive law, which we describe next. 

We employ a \emph{piecewise constant estimation} scheme~\citep{hovakimyan2010L1} based on the solution of~\Eqref{eq:prediction_error_dynamics}, which can be expressed using the following equation: 
\begin{equation}\label{eq:prediction_error_solution}
    \Tilde{x}(t) = \exp{(A_st)}\tilde{x}(0) + \int_{0}^{t} \exp{\left(A_s(t-\tau)\right)}(\hat{\sigma}(\tau)-d(\tau,x(\tau), u(\tau))d\tau.
\end{equation}
For a given sampling time $T_s > 0$, we use a piecewise constant estimate defined as
\[
\hat{\sigma}(t) = \hat{\sigma}(iT_s), \quad t \in [iT_s,(i+1)T_s), ~i \in \mathbb{N} \cup \{ 0\} . 
\]
% We similarly set $\bar{h}(x(t)) = \bar{h}(x(iT_s))$.
When the system is initialized ($i=0$), we set $\tilde{x}(t) =0$ which implies $\hat{\sigma}(t) = 0$ for $t \in [0,T_s)$. Now consider a time index $i\in \mathbb{N}$, and the time interval $[iT_s,(i+1)T_s)$. Over this time interval, the solution of~\Eqref{eq:prediction_error_solution}, obtained by applying the piecewise constant representation, can be written as
\begin{equation*}\label{eq:prediction_error_solution_2}
\begin{aligned}
    \Tilde{x}(&t) = \exp{(A_s (t - iT_s))}\tilde{x}(iT_s) \\
    &+ \int_{iT_s}^t \exp{\left(A_s(t-\tau)\right)}(\hat{\sigma}(\tau)- d(\tau,x(\tau), u(\tau))d\tau, \quad t \in [iT_s,(i+1)T_s).
\end{aligned}
\end{equation*}
Now, note that over the previous interval $t \in [(i-1)T_s,iT_s)$ the system was affected by uncertainty $d(t,x(t), u(t))$, which  resulted in $\tilde{x}(iT_s) \neq 0$. At the end of the time interval $[iT_s,(i+1)T_s)$, we obtain
\begin{align}
\Tilde{x}((i+1)T_s) =& \exp{(A_sT_s)}\tilde{x}(iT_s) + \int_{iT_s}^{(i+1)T_s} \exp{\left(A_s( (i+1)T_s -\tau)\right)}\hat{\sigma}(iT_s)d\tau +R_{(i+1)T_s} ~\notag\\
=& \exp{(A_s T_s)}\tilde{x}(iT_s) +  A_s^{-1} \left(\exp{\left(A_s T_s\right)} - \mathbb{I}_n\right)\hat{\sigma}(iT_s) +R_{(i+1)T_s} ~\label{eq:pw_solution},
\end{align} 
where
$$R_{(i+1)T_s} \triangleq - \int_{iT_s}^{(i+1)T_s} \exp{\left(A_s( (i+1)T_s -\tau)\right)}d(\tau,x(\tau), u(\tau))d\tau $$
is the residual term that captures the uncertainty entered during $[iT_s,(i+1)T_s)$.

Let
\begin{align}\label{eq:adaptation_law}
    \hat{\sigma}(iT_s) = -\Phi^{-1}(T_s)\mu(iT_s),\quad \hat{\sigma}(0)=0,
\end{align} 
where $\Phi(T_s) = A_s^{-1}(\exp{(A_sT_s)}-\mathbb{I})$ and $\mu(iT_s) = \exp{(A_sT_s)}\Tilde{x}(iT_s)$ for $i = 0,1,2, \cdots$. Substituting this into~\Eqref{eq:pw_solution} removes the first two terms, leaving us only with the residual term, which will appear as the initial condition of the next time interval. In other words, the adaptation law attempts to remove the effect of the uncertainty introduced in the current time interval by addressing it at the start of the subsequent interval. Interested readers can refer to~\cite[Ch. 2]{kharisov20131} for further details on the piecewise constant adaptive law.

Note that a small sampling time $T_s$ results in a small prediction error $\|\Tilde{x}(iT_s)\|$ for each $i=1,2,\cdots$. Therefore, it is desirable to keep $T_s$ small up to the hardware limit. However, setting a small $T_s$ and/or having large eigenvalues of $A_s$ can lead to a high adaptation gain ($\Phi^{-1}$ in~\Eqref{eq:adaptation_law}). This can result in high-frequency uncertainty estimation, which can reduce the robustness of the controlled system if we directly apply $u_a(t)=-\hat{\sigma}_m(t)$ to cancel the estimated matched uncertainty. Therefore, we use a low-pass filter in the controller to decouple the fast estimation from the control loop, allowing us to employ an arbitrarily fast adaptation while maintaining the desired robustness. To be specifc, the input $u_a(t)$ is given by
\begin{equation*}
    u_a(s) = -C(s)\mathfrak{L}\left[\hat{\sigma}_m(t) \right],
\end{equation*}
where $\mathfrak{L}[\cdot]$ denotes the Laplace transform, and the $C(s)$ is the  low-pass filter. The cutoff frequency of the low-pass filter is chosen to satisfy small-gain stability conditions, examples of which can be found in~\citep{wang2011,lakshmanan2020safe}. We refer the interested reader to~\citep{pravitra2020,wu20221,cheng2022improving} for further reading on the design process of the \ellone adaptive control.

\section{Remarks on  Assumption~\ref{assumption}}\label{appendix:assumption}
It is evident from~\Eqref{eq:affinzation} that our method relies on the continuous differentiability of $\hat{F}_\theta(x(t),u(t))$ with respect to $u(t)$ to ensure the continuity of $\hat{F}_\theta^{a}(x(t),u(t))$. Such a requirement is readily satisfied when using $\mathcal{C}^{1}$ (or higher order continuously differentiable) activation functions for $\hat{F}_\theta(x(t),u(t))$, such as \texttt{sigmoid}, \texttt{tanh}, or \texttt{swish}. For MBRL algorithms that use activation functions that are not $\mathcal{C}^1$, we can skip the switching law~(\Eqref{eq:switching}) to avoid making affine approximations at non-differentiable points (e.g., the origin for \texttt{ReLU}). 

The first part of Assumption~\ref{assumption} on the regularity of $F(x(t),u(t))$ and $W(t,x(t))$ is standard to ensure the well-posedness (uniqueness and existence of solutions) for~\Eqref{eq:ground_truth_nonlinear}~\cite[Theorem~3.1]{khalil2002nonlinear}. Furthermore, as stated above, since $\hat{F}_{\theta}(x(t),u(t))$ is $\mathcal{C}^1$ in its arguments,  its derivative $\hat{F}_{\theta}(x(t),u(t))$ is continuous and hence, satisfies the local Lipschitz continuity over the compact sets $\mathcal{X}$ and $\mathcal{U}$ trivially. 

The assumption in~\Eqref{eq:nn_model_error} is satisfied owing to the Lipschitz continuity of the function in its respective arguments. If the bound is unknown, it is possible to collect data by interacting with the environment under a specific policy and initial condition, and then compute a probabilistic bound. However, such a bound is applicable only to the chosen data set and may not hold for other choices of samples, which is the well-known issue of distribution shift~\citep{quinonero2008dataset}. This assumption, which explicitly bounds the unknown component of the dynamics, although cannot be guaranteed in the real environment, is commonly made in assessing theoretical guarantees of error propagation when using learned models, as seen in previous papers~\citep{knuth2021planning,manzano2020robust,koller2018learning}.

\section{Proof of Theorem~\ref{theorem}}\label{appendix:theoretical_analysis}

\begin{proof}[Proof of Theorem~\ref{theorem}.]
We begin by applying the triangle inequality for the estimation error:
\begin{equation}\label{eq:theorem1_setup}
\begin{aligned}
     &\|e(t,x(t),u(t))\| \\
     &= \|F(x(t),u(t))+W(t,x(t), u(t))-G_\theta(x(t))-H_\theta(x(t))u(t)-\hat{\sigma}(t)\| \\
     &= \|\hat{F}_\theta(x,u(t))+l(t,x(t),u(t))-\hat{F}^a_\theta(x(t),u(t))-\hat{\sigma}(t)\|\\
    &\leq \|\hat{F}_\theta(x(t),u(t))-\hat{F}^a_\theta(x(t),u(t))\| + \|l(t,x(t),u(t))-\hat{\sigma}(t)\| \\
    &\leq  \epsilon_a   + \|l(t,x(t),u(t))-\hat{\sigma}(t)\|, \quad i \in \{0\} \cup \mathbb{N},
\end{aligned}    
\end{equation}
where we used~\Eqref{eq:switching}. 

For the case when $i=0$, due to the initial conditions $\Tilde{x}(0)=0$, $\hat{\sigma}(0)=0$, and the assumption in~\Eqref{eq:nn_model_error}, we get that
    \begin{align*}
        \|l(t,x(t),u(t))-\hat{\sigma}(0)\| = \|l(t,x(t),u(t))\| \leq \epsilon_l, \quad \forall t \in [0,T_s).
    \end{align*} Substituting this expression into~\Eqref{eq:theorem1_setup} proves the stated result for $t \in [0,T_s)$.

Next, we bound the term $\|l(t,x(t),u(t))-\hat{\sigma}(iT_s)\|$ in~\Eqref{eq:theorem1_setup} for all $t \in [T_s,t_{\max})$. Consider an $i \in \mathbb{N}$ that corresponds to $ t \in [T_s,t_{\max})$, i.e., $i \in \left\{1,\dots, \lfloor t_{\max}/T_s  \rfloor  \right\} \triangleq \mathcal{I}$. For any such $i$, substituting the adaptation law from~\Eqref{eq:adaptation_law} into~\Eqref{eq:pw_solution} for the interval $[(i-1)T_s,iT_s)$ produces the following expression
    \begin{equation}\label{eq:xtilde_general}
        \Tilde{x}(iT_s) = -\int_{(i-1)T_s}^{iT_s} \exp(A_s(iT_s-\tau))d(\tau,x(\tau), u(\tau))d\tau, \quad \forall i \in \mathcal{I}.
    \end{equation}
    Replacing $d(\tau,x(\tau), u(\tau))= \begin{bmatrix} d_1(\tau, x(\tau), u(\tau)) & \dots & d_n(\tau, x(\tau), u(\tau))\end{bmatrix}^\top $ in \Eqref{eq:xtilde_general} and, by the definition of $A_s$ in the theorem statement, we obtain
    %%%%%%%%%%%%%%%%%%
    \begin{align*}
        \Tilde{x}(iT_s)=-\bigints_{(i-1)T_s}^{iT_s}
        \left[\begin{array}{c}
             \exp(\lambda_1(iT_s-\tau))d_1(\tau,x(\tau), u(\tau))\\
             \vdots\\
             \exp(\lambda_n(iT_s-\tau))d_n(\tau,x(\tau), u(\tau))
        \end{array}\right] d\tau,
    \end{align*}
    %%%%%%%%%%%%%%%%%%
    for all $i \in \mathcal{I}$.
    For brevity, we denote $d_j(\tau)=d_j(\tau,x(\tau),u(\tau))$ for $j=1,2,\ldots,n$.
    
    \newpage
    Since $d_j(t)$ is continuous due to Assumption~\ref{assumption} and $\exp(A_s(iT_s-\tau))$ is positive semi-definite, we invoke the Mean Value Theorem~\cite[Sec.~24.1]{mendelson2022schaum} element wise.  We conclude that there exist $t_{c_j}\in [(i-1)T_s,iT_s)$ for each $j\in \{1,\dots,n\}$ such that 
    \begin{equation}\label{eq:xtilde_vector}
        \begin{aligned}
        \Tilde{x}(iT_s) &= -\left[\begin{array}{c}
             \int_{(i-1)T_s}^{iT_s}\exp(\lambda_1(iT_s-\tau)) d_1(t_{c_1})d\tau \\
             \vdots\\             
             \int_{(i-1)T_s}^{iT_s}\exp(\lambda_n(iT_s-\tau))d_n(t_{c_n})d\tau
        \end{array}\right] \\
        &=\left[\begin{array}{c}
             \frac{1}{\lambda_1}(1-\exp(\lambda_1 T_s))d_1(t_{c_1})  \\
            \vdots\\
            \frac{1}{\lambda_n}(1-\exp(\lambda_n T_s))d_n(t_{c_n})
        \end{array}\right].
    \end{aligned}
    \end{equation}
     Substituting ~\Eqref{eq:xtilde_vector} into~\Eqref{eq:adaptation_law} gives
    \begin{equation}\label{eq:hbar_sigmahat}
        \hat{\sigma}(t)= \hat{\sigma}(iT_s) =\left[\begin{array}{c}
            \exp(\lambda_1T_s)d_1(t_{c_1})\\
            \vdots\\
            \exp(\lambda_nT_s)d_n(t_{c_n})
        \end{array}\right],
    \end{equation}
    which is the piece-wise constant estimate of the uncertainty for $t\in [iT_s,(i+1)T_s)$.
    
    Next, using the piece-wise constant adaptation law from~\Eqref{eq:hbar_sigmahat}, we bound  $\|l(t,x(t),u(t)) -\hat{\sigma}(t)\|$, for $t\in [iT_s,(i+1)T_s)$. Since $A_s$ is Hurwitz and diagonal, its diagonal elements satisfy $\lambda_j<0$, for $j=1,\ldots,n$. Thus, we have
    %%%%%%%%%%%%%%%%%%%%%%%%%%%%%%%%%%%%%%%%%%%%%
    \begin{align}\label{eq:Thm1:Inter:1}
        \|l(t,x(t),u(t)) -\hat{\sigma}(t)\| &= \left\|l(t,x(t),u(t)) -\left[\begin{array}{c}
         \exp(\lambda_1T_s)d_1(t_{c_1}) \\
        \vdots\\
        \exp(\lambda_nT_s)d_n(t_{c_n})
        \end{array}\right]\right\| \notag \\
        &= \|l(t,x(t),u(t)) - d(t_c) + (\mathbb{I}_n-\exp{(A_sT_s)}d(t_c))\| \notag 
        \\
        &\leq \|l(t,x(t),u(t)) -d(t_c)\| + \|\mathbb{I}_n-\exp{(A_sT_s)}\|\|d(t_c)\| \notag 
        \\
        &\leq \|l(t,x(t),u(t)) -d(t_c)\|+(1-\exp(\lambda_{\min}T_s))\|d(t_c)\|,
    \end{align}
    %%%%%%%%%%%%%%%%%%%%%%%%%%%%%%%%%%%%%%%%%%%%%%%
    where $d(t_c) := \begin{bmatrix} d_1(t_{c_1}) & \dots & d_n(t_{c_n})\end{bmatrix}^\top$, $d(\tau) = l(\tau, x(\tau), u(\tau)) +a(x(\tau),u(\tau))$, and $\lambda_{\min}$ denotes the eigenvalue of $A_s$ that has the minimum absolute value. Using the triangle inequality, we get $\|d(t_c)\| = \|l(t_c,x(t_c),u(t_c) + a(x(t_c),u(t_c))\| \leq \epsilon_l+\epsilon_a$. Thus,~\Eqref{eq:Thm1:Inter:1} can be written as
    %%%%%%%%%%%%%%%%%%%%%%%%%%%%%%%%%%%%%%%%%%%%%
    \begin{align}\label{eq:Thm1:Inter:2}
        \|l(t,x(t),u(t)) -\hat{\sigma}(t)\| 
        &\leq \|l(t,x(t),u(t)) -d(t_c)\|+(1-\exp(\lambda_{\min}T_s))(\epsilon_l+\epsilon_a).
    \end{align}
    %%%%%%%%%%%%%%%%%%%%%%%%%%%%%%%%%%%%%%%%%%%%%%%
    Next, we obtain the upper bound for $\|l(t,x(t),u(t)) -d(t_c)\|$ as
    \begin{align}\label{eq:Thm1:Inter:3}
        \|l(t,x(t),u(t)) -d(t_c)\| &=\|l(t,x(t),u(t)) -l(t_c,x(t_c),u(t_c)) - a(x(t_c), u(t_c))\| \notag 
        \\
        &\leq \|l(t,x(t),u(t)) - l(t_c,x(t_c),u(t_c))\| + \norm{a(x(t_c), u(t_c))} \notag 
        \\       
        &\leq \|l(t,x(t),u(t)) - l(t_c,x(t_c),u(t_c))\| +\epsilon_a,
    \end{align}
    where we used the triangle inequality and~\Eqref{eq:switching}. Due to the Assumption~\ref{assumption}, $l(t,x(t),u(t))$ is Lipschitz over the domain of its arguments. Hence, there exist positive scalars $L_{l,t}, L_{l,x}, L_{l,u}$ such that 
    \begin{equation}\label{eq:learning_error_bound}
        \|l(t,x(t),u(t)) - l(t_c,x(t_c),u(t_c))\|\leq L_{l,t}|t-t_c| + L_{l,x}\|x(t)-x(t_c)\| + L_{l,u}\|u(t)-u(t_c)\|.
    \end{equation}
    Furthermore, due to the compactness of  $\mathcal{X}$ and $\mathcal{U}$, there exist $L_{x,t},L_{u,t}$ such that the following inequalities hold:
    \begin{align*}
      \|x(t)-x(t_c)\| &\leq L_{x,t} |t-t_c|, \quad  \|u(t)-u(t_c)\| \leq L_{u,t}|t-t_c|.
    \end{align*}
    Substituting these bounds into~\Eqref{eq:learning_error_bound}, we get
    \begin{equation}\label{eq:Thm1:Inter:4}
        \|l(t,x(t),u(t)) - l(t_c,x(t_c),u(t_c))\|\leq L \left|t-t_c\right| \leq LT_s,
    \end{equation}
     where $L\triangleq L_{l,t}+L_{l,x}L_{x,t} + L_{l,u}L_{u,t} < \infty$.
     We proceed by sequentially applying the derived bounds, starting with the substitution of~\Eqref{eq:Thm1:Inter:4} into~\Eqref{eq:Thm1:Inter:3}, and then employing the resulting bound in~\Eqref{eq:Thm1:Inter:2}. The proof is then concluded by incorporating the final bound into~\Eqref{eq:theorem1_setup} and noting that
    \begin{equation*}
        (1-\exp(\lambda_{\min}T_s))(\epsilon_l+\epsilon_a) + LT_s \in \mathcal{O}(T_s).
    \end{equation*}
\end{proof}
\begin{remark}\label{rem:Thm1}
We provide some insights into the interpretation of this theorem. The theorem serves to quantify the predictive quality of the state predictor in the \ellone add-on scheme in terms of the model approximation errors $\epsilon_l$ and $\epsilon_a$, and the parameters governing the \ellone add-on scheme ($T_s$, $\lambda_{\min}$). As the control input is computed by low-pass filtering the uncertainty estimate, the performance of the \ellone augmentation is inherently tied to its predictive quality. Theorem~\ref{theorem} establishes that the error in state prediction, induced by estimated uncertainty, can be reduced down to $2\epsilon_a$ by reducing the sampling time $T_s$. In other words, we can accurately estimate the learning error $l(t, x(t), u(t))$, with the predictive accuracy being bounded only by the tunable parameter $\epsilon_a$.
\end{remark}

% \newpage
\section{Extended Simulation Results}\label{sec:extended_simulation}
\subsection{Experiment setup}
We provide the dimensionality of the selected environments for our simulation analysis in Table~\ref{tab:environment_dimension}. For \ellone-METRPO, the number of iterations for each environment was chosen to obtain asymptotic performance, whereas for \ellone-MBMF we fixed the number of iterations to 200K. Such a setup for MBMF is due to the unique structure of MBMF to use MBRL only to serve as the policy initialization, with which MFRL is executed until the performance reaches the asymptotic results. See Appendix~\ref{appendix:mbmf} for a detailed explanation of MBMF.

% In MBMF, the policy trained using data generated by the learned model-based controller serves as the initial policy for the model-free reinforcement learning algorithm and is designed to run for 1 million time steps. Consequently, our evaluation focused solely on the initial phase of training, and hence, we did not assess the framework's asymptotic performance since it was primarily designed for the model-based component of the algorithm.

%%%%%%%%%%%%%%%%%%%%%%%%%%%
\begin{table}[h]
    \centering
    \caption{Dimensions of state and action space of the environments used in the simulations.}
    \begin{tabular}{c|c|c}
    \toprule
         Environment Name & State space dimension $(n)$ & Action space dimension $(m)$ \\
         \midrule
         Inverted Pendulum& 4&1\\
         Swimmer & 8 & 2\\
         Hopper & 11 & 3\\
         Walker & 17 & 6\\
         Halfcheetah & 17 & 6\\
    \bottomrule
    \end{tabular}
    % \vspace{5pt}
    \label{tab:environment_dimension}
\end{table}
% \vspace{-15pt}
%%%%%%%%%%%%%%%%%%%%%%%
We adopted the hyperparameters that have been reported to be effective by the baseline MBRL, which in our case are  METRPO and MBMF~\citep[Appendix B.4, B.5]{wang2019benchmarking}. Additional hyperparameters introduced by the \lmbrl scheme are the affinization threshold $\epsilon$, the cutoff frequency $\omega$, and the Hurwitz matrix $A_s$. Throughout all experiments, we fixed $A_s$ as a negative identity matrix $-\mathbb{I}_n$. For the Inverted Pendulum environment, we set $\epsilon=1$ and for Halfcheetah $\epsilon=3$, while for other environments, we chose $\epsilon=0.3$. Additionally, we selected a cutoff frequency of $\omega = 0.35/T_s$, where $T_s$ represents the sampling time interval of the environment. {\bf It is important to note that the \ellone controller has not been redesigned or retuned through all the experiments.}

%\begin{remark}
    % Ideally, the filter bandwidth should be chosen to satisfy the small-gain stability conditions as elaborated in~\cite{wang2011}. However, finding a suitable Lyapunov function is a challenging task in general for nonlinear functions. Furthermore, it is impractical for our setup as it would require repeated calculations each time the system is affinized. Consequently, we employed a fixed bandwidth throughout the episodes and environments, as a practical alternative while trading in the explicit stability guarantee.
    
%The purpose of the low-pass filter $C(s)$ is to attenuate the high-frequency components in the adaptive control signals. By selecting a low cutoff frequency $\omega$ for the filter, the robustness of the adaptive controller can be systematically adjusted. Importantly, it is worth noting that the \ellone control augmentation system does not necessitate any modifications to the filter design or adaptation rate in order to accommodate changes in the system dynamics.
%\end{remark}

\subsection{Technical Remarks}
If the baseline algorithm employs any data processing techniques such as input/output normalization, as discussed briefly in Section~\ref{MBRL}, our state predictor and controller~(\Eqref{eqn:predictor},\Eqref{eqn:control_input}) must also follow the corresponding process. 

\subsubsection{METRPO}
METRPO trains an ensemble model, from which fictitious samples are generated. Then, the policy network is updated following the TRPO~\citep{schulman2015trust} in the policy improvement step. The input and output of the neural network are normalized during the training step, and consequently, calculation of the Jacobian in \ellone-METRPO must \emph{unnormalize} the result. Specifically, this process is carried out by applying the chain rule, which includes multiplying the normalized Jacobian matrix ($J'$) by the standard deviations of the inputs and outputs given by~\Eqref{eq:normalization}:
\begin{equation}\label{eq:normalization}
   J= D_{\Delta x} J' D_{x,u}^{-1}  ,
\end{equation} 
where $D_{\Delta x} = \texttt{diag}\{\sigma_{\Delta x_1}, \ldots,\sigma_{\Delta x_n}\}$ and $D_{x,u} = \texttt{diag}\{\sigma_{x_1}, \ldots, \sigma_{x_n},\sigma_{u_1}, \ldots,\sigma_{u_m}\}.$
% $D_{\Delta s} =
%   \begin{bmatrix}
%     \sigma_{\Delta s_1} & & \\
%     & \ddots & \\
%     & &  \sigma_{\Delta s_n}
%   \end{bmatrix}$ and
% $D_{s,a} =
%   \begin{bmatrix}
%       \sigma_{s_1} &  &  &  &  &  \\
%           & \ddots  &  &   &  &   \\
%          &  & \sigma_{s_n}& & &  \\
%           &  &  & \sigma_{a_1}  &  &   \\
%          &  &  &  & \ddots  &    \\
%           &  &  & &   &\sigma_{a_m}
%   \end{bmatrix}$

This unnormalized Jacobian ($J$)  is subsequently utilized to generate the \ellone adaptive control output.

\subsubsection{MBMF}\label{appendix:mbmf}
In the MBMF algorithm~\citep{nagabandi2018neural}, the authors begin by training a Random Shooting (RS) controller. This controller is then distilled into a neural network policy using the supervised framework DAgger~\citep{ross2011reduction}, which minimizes the KL divergence loss between the neural network policy and the RS controller. Then, the policy is fine-tuned using standard model-free algorithms like TRPO~\citep{schulman2015trust} or PPO~\citep{schulman2017proximal}. We adopt a similar approach to what was done for METRPO. The Jacobian matrix of the neural network is unnormalized based on~\Eqref{eq:normalization}. The adaptive controller is augmented to the RS controller based on the latest model trained. 

\subsection{Experiment results}\label{sec:extended-result}
In this section, we first present the results of \ellone-MBMF in comparison to MBMF without \ellone augmentation. The corresponding tabular results are summarized in Table~\ref{table:l1-mbmf}. Noticeably, \ellone augmentation improves the MBMF algorithm in \emph{every} case uniformly.
%%%%%%%%%%%%%%%%%%%%%%%
\begin{table*}[h]
  \scriptsize
  \centering
  \caption{Performance comparison between MBMF and $\mathcal{L}_1$-MBMF (Ours). The performance is averaged across multiple random seeds with a window size of 5000 timesteps at the end of the training. Higher performance is written in bold and green.}
  \label{table:l1-mbmf}
  \vspace*{-0.5em}
  \addtolength{\tabcolsep}{-2pt}
  \begin{center}
  \resizebox{\columnwidth}{!}{%
  \begin{tabular}{c|cc|cc|cc}
    \toprule
	\textbf{ } & 
	\multicolumn{2}{c|}{\bfseries \footnotesize Noise-free} & 
	\multicolumn{2}{c|}{\bfseries \footnotesize $\mathbf{\sigma_a=0.1}$} &
	\multicolumn{2}{c}{\bfseries \footnotesize $\mathbf{\sigma_o=0.1}$}\\ 
	\textbf{\footnotesize Env.} 
	
	& {\scriptsize MB-MF} & {\scriptsize $\mathcal{L}_1$-MB-MF} & {\scriptsize MB-MF} & {\scriptsize $\mathcal{L}_1$-MB-MF} & {\scriptsize MB-MF}& {\scriptsize $\mathcal{L}_1$-MB-MF}
	\\ \toprule

	{\footnotesize Inv. P.}
    & $-100.5\pm  4.3 $      &\textcolor{applegreen}{$\mathbf{-10.5 \pm 3.7}$}    &  $-7.4 \pm 1.5$     &   \textcolor{applegreen}{$\mathbf{-4.8 \pm 1.9}$}   &  $-10.2 \pm  2.4$      & \textcolor{applegreen}{$\mathbf{-5.09 \pm 1.6}$} \\ \midrule

 	{\footnotesize Swimmer}
 	& $284.9\pm  25.1$     & \textcolor{applegreen}{$\mathbf{314.3 \pm 3.3}$}    & $304.8\pm  1.9$    & \textcolor{applegreen}{$\mathbf{314.5 \pm 0.6}$}      &  $292.8\pm  1.3$  & \textcolor{applegreen}{$\mathbf{294.3 \pm 4.3}$}    \\ \midrule
 	
    {\footnotesize Hopper}
 	& $-1047.4  \pm   1098.7 $     & \textcolor{applegreen}{$\mathbf{350.1 \pm 465.2}$}     &$-877.9 \pm   383.4 $     &  \textcolor{applegreen}{$\mathbf{-285.4 \pm 65.3}$}    & $-996.9 \pm   206.0 $    & \textcolor{applegreen}{$\mathbf{-171.5  \pm 317.3 }$}  \\ \midrule
 	
    {\footnotesize Walker}
 	& $-1743.7  \pm 233.3$     & \textcolor{applegreen}{$\mathbf{-1481.7 \pm 322.9}$}      & $-2962.2 \pm 178.6$    & \textcolor{applegreen}{$\mathbf{-2447.4 \pm 329.7}$}  & $-3348.8 \pm 210.1$   & \textcolor{applegreen}{$\mathbf{-2261.4 \pm 381}$} \\ \midrule
 	
    {\scriptsize Halfcheetah}
 	& ${126.9\pm 72.7}$   &\textcolor{applegreen}{$\mathbf{304.5 \pm 56.0}$}   & $184.0 \pm 148.9$     & \textcolor{applegreen}{$\mathbf{299.8 \pm 61.0}$}   & $146.1\pm 87.8$     &\textcolor{applegreen}{$\mathbf{235.2 \pm 19.2}$} \\

    \bottomrule
	\end{tabular}}
	\end{center}
    \vspace*{-0.5em}
\end{table*}
%%%%%%%%%%%%%%%%%%%%%%%%%%%%%

Additionally, we provide detailed tabular values corresponding to the results shown in Fig.~\ref{fig:ablation_barplot}. Table~\ref{table:ablation} provides a summary of the scenarios where \ellone control is augmented only during either training or testing. The application of \ellone control during the testing phase clearly benefits from the explicit rejection of system uncertainty, leading to performance improvement. On the contrary, when \ellone control is applied during the training phase, it not only mitigates uncertainty along the trajectory but also implicitly affects the training process by inducing a shift in the distribution of the training dataset. This study compares these two types of impact brought about by the \ellone augmentation.

%%%%%%%%%%%%%%%%%%%%%%%%%
{\renewcommand{\arraystretch}{1.3}
\begin{table*}[h]
  \scriptsize
  \centering
  \caption{Comparison of \ellone augmentation effects during training and testing. \ellone-METRPO (Train) refers to the application of \ellone augmentation solely during training, whereas \ellone-METRPO (Test) indicates training without \ellone augmentation and the application of \ellone only during testing.}
  \label{table:ablation}
  \vspace*{-0.5em}
  \addtolength{\tabcolsep}{-2pt}
  \begin{center}
  \resizebox{\columnwidth}{!}{%
  \begin{tabular}{c|cc|cc|cc}
    \toprule
	\textbf{ } & 
	\multicolumn{2}{c|}{\bfseries \footnotesize Noise-free} & 
	\multicolumn{2}{c|}{\bfseries \footnotesize $\mathbf{\sigma_a=0.1}$} &
	\multicolumn{2}{c}{\bfseries \footnotesize $\mathbf{\sigma_o=0.1}$}\\ 
	\textbf{\footnotesize Env.} 
	
	& {\scriptsize $\mathcal{L}_1$-METRPO (Train)} & {\scriptsize $\mathcal{L}_1$-METRPO(Test)} & {\scriptsize $\mathcal{L}_1$-METRPO(Train)} & {\scriptsize $\mathcal{L}_1$-METRPO(Test)} & {\scriptsize $\mathcal{L}_1$-METRPO(Train)}& {\scriptsize $\mathcal{L}_1$-METRPO(Test)}
	\\ \toprule

	{\footnotesize Inv. P.}
    & \textcolor{applegreen}{$\mathbf{-8.50\pm 20.75}$}       &  $-19.36\pm 22.3$   & \textcolor{applegreen}{$\mathbf{-3.52\pm 8.08}$}      & $-49.72\pm 47.34$   & $-41.63 \pm 19.11$     & \textcolor{applegreen}{$\mathbf{-37.00\pm 39.07}$}  \\ \midrule

 	{\footnotesize Swimmer}
 	& $332.6\pm 1.3$     & $332.6\pm 1.6$     & \textcolor{applegreen}{$\mathbf{321.8\pm 1.0}$}     & $298.9\pm 3.1$    & $32.9 \pm 1.5$     & \textcolor{applegreen}{$\mathbf{52.0 \pm 8.7}$}  \\ \midrule
 	
    {\footnotesize Hopper}
 	& $1201.2  \pm 90.8$     & \textcolor{applegreen}{$\mathbf{1269.9 \pm 752.9}$}     & $771.1 \pm49.8$     & \textcolor{applegreen}{$\mathbf{818.1 \pm 394.2}$}    & \textcolor{applegreen}{$\mathbf{-931.7 \pm 15.4}$}     & $-976.8 \pm 73.1$  \\ \midrule
 	
    {\footnotesize Walker}
 	& $-7.0 \pm 0.1$     & \textcolor{applegreen}{$\mathbf{-5.9 \pm 0.0}$}     & \textcolor{applegreen}{$\mathbf{-6.5 \pm 0.3}$}     & $-7.5\pm 0.2$    & \textcolor{applegreen}{$\mathbf{-6.3 \pm 0.0}$}     & $-10.4 \pm 0.2$  \\ \midrule
 	
    {\scriptsize Halfcheetah}
 	& \textcolor{applegreen}{$\mathbf{2706.2\pm 1170.4}$}     & $1921.56\pm 821.34$     & $1834 \pm 434.87$     & \textcolor{applegreen}{$\mathbf{1957.5\pm 581.6}$}     & $987.90\pm 435.90$     & \textcolor{applegreen}{$\mathbf{1022.1\pm 619.8}$} \\

    \bottomrule
	\end{tabular}}
	\end{center}
    \vspace*{-0.5em}
\end{table*}
}
%%%%%%%%%%%%%%%%%%

Notably, there is no consistent trend regarding whether \ellone control has a greater impact during testing or training phases. The primary conclusion drawn from this ablation study - in conjunction with Fig.~\ref{fig:ablation_barplot} - is that \ellone augmentation yields the greatest benefits when applied to both training and testing. One possible explanation for this observation is that such consistent augmentation avoids a shift in the policy distribution, leading to desired performance.

Next, in Fig.~\ref{fig:learning_curve}, we report the learning curves of the main result.
\newpage

\begin{figure}[h!]
    \centering
    \includegraphics[width = \textwidth]{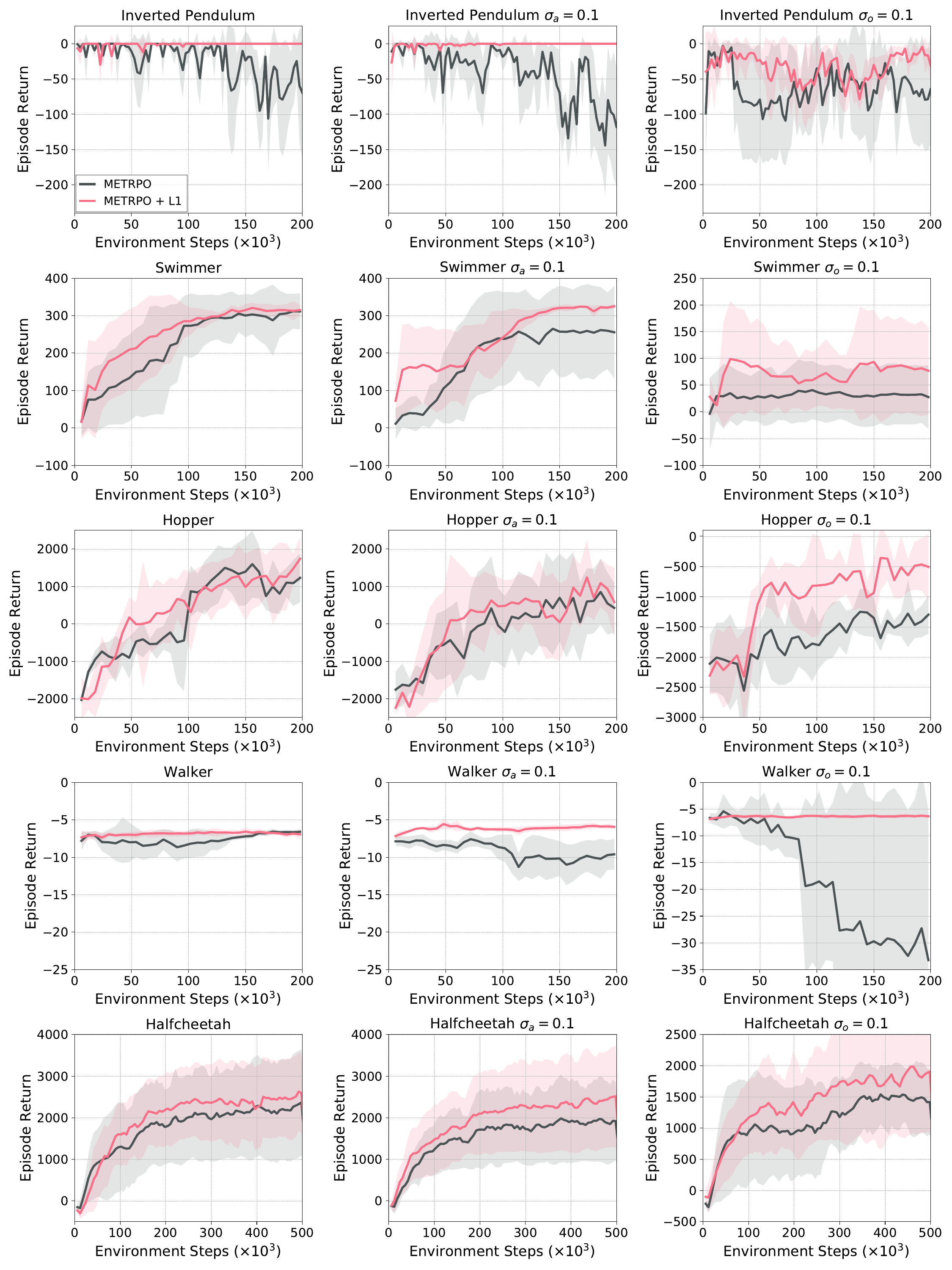}
    \caption{Plots of \ellone-METRPO learning curves as a function of episodic steps. The performance is averaged across multiple random seeds such that the solid lines indicate the average return at the corresponding timestep, and the shaded regions indicate one standard deviation.}
    \label{fig:learning_curve}
\end{figure}
\newpage

\begin{figure}[h!]
    \centering
    \includegraphics[width = \textwidth]{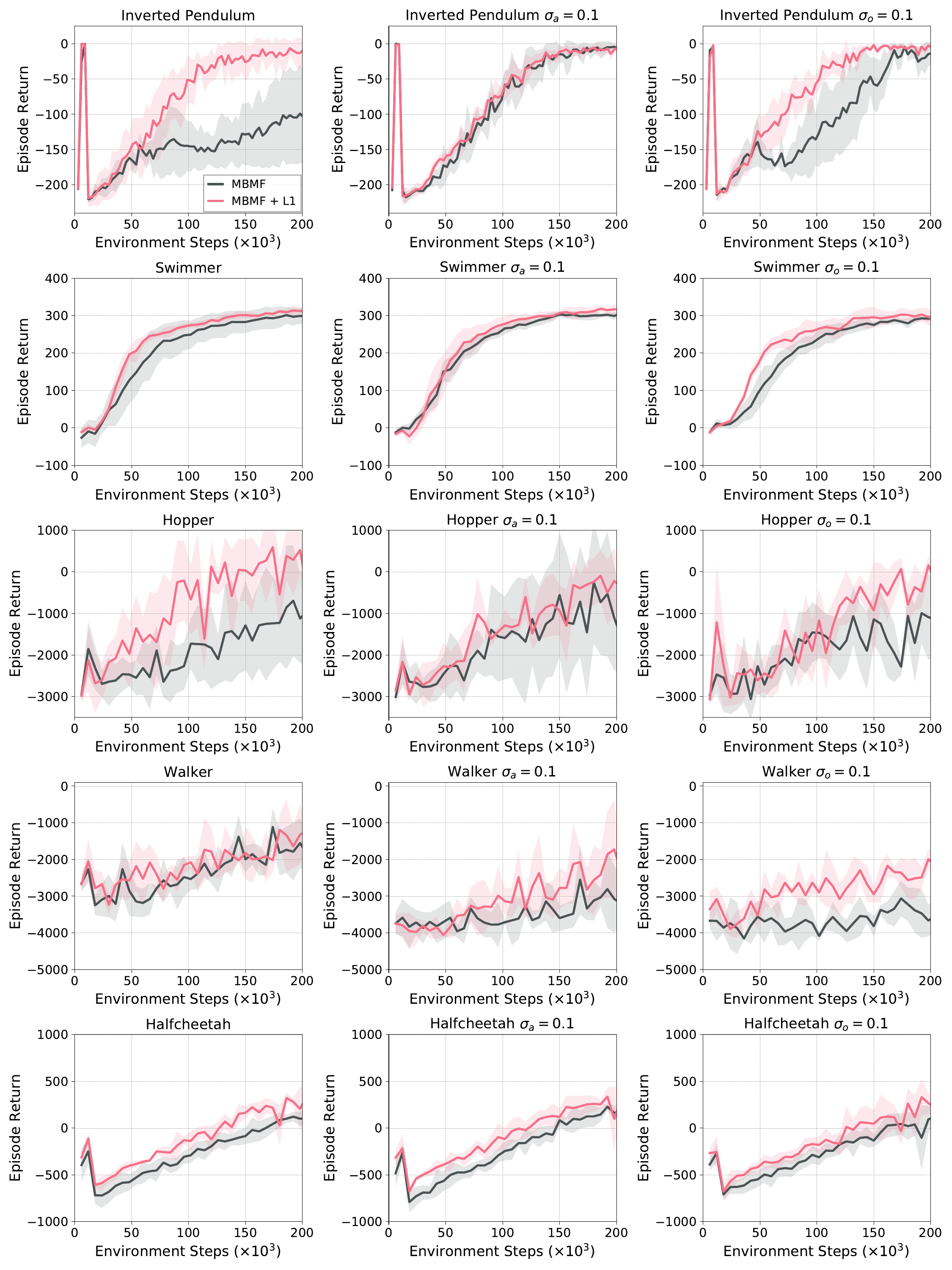}
    \caption{Plots of \ellone-MBMF learning curves as a function of episodic steps. The evaluation of the performance is identical to \ellone-METRPO. } 
    \label{fig:MBMF_learning_curve}
\end{figure}
%%%%%%%%%%%%%%%%%
\newpage
\subsection{Comparison with Probabilistic Models}\label{sec:prob-model}
\begin{figure}[h!]
    \centering
    \includegraphics[width = 0.5\textwidth]{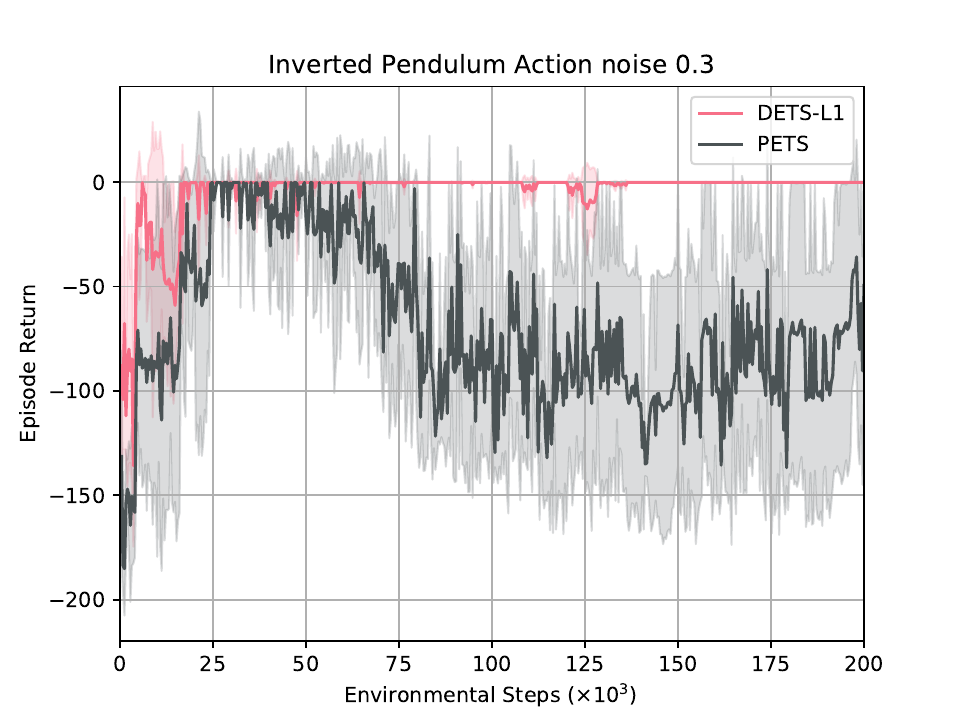}
    \caption{Plots of \ellone-DETS vs PETS learning curves as a function of episodic steps. } 
    \label{fig:dets_learning_curve}
\end{figure}

Probabilistic models, as discussed in ~\citep{chua2018deep, Wang2020Exploring}, offer a common approach in Reinforcement Learning (RL) to tackle model uncertainty. In contrast, our approach, centered on a robust controller, shares a similar spirit but differs in architecture. While previous works directly integrate uncertainty into decision-making, for example, through methods like sampling-based Model Predictive Control (MPC)~\citep{chua2018deep}, our approach takes a unique path by decoupling the process. We address uncertainty by explicitly estimating and mitigating it based on the learned deterministic nominal dynamics, allowing the MBRL algorithm to operate as intended.

Recently, the authors in~\citep{zheng2022model} emphasized that the empirical success of probabilistic dynamic model ensembles is attributed to their Lipschitz-regularizing aspect on the value functions. This observation led to the hypothesis that the ensemble's key functionality is to regularize the Lipschitz constant of the value function, not in its probabilistic formulation. The authors have shown that the predictive quality of deterministic models does not show much difference with probabilistic (ensemble) models, leading to the conclusion that deterministic models can offer computational efficiency and practicality for many MBRL scenarios. In this context, our work exploits the practical advantages of using deterministic models, while $\mathcal{L}_1$ adaptive controller accounts for the randomness present in the environment.

In Fig.~\ref{fig:dets_learning_curve}, we conducted supplementary experiments comparing PETS and its deterministic counterpart, DETS, with \ellone augmentation. The results were obtained using multiple random seeds and 200,000 timesteps in the Inverted Pendulum environment with an action noise of $\sigma_a = 0.3$, demonstrating that the deterministic model with \ellone augmentation (\ellone-DETS) can outperform the probabilistic model approach (PETS). However, it's important to note that this comparison is specific to one environment. We refrain from making broad claims regarding DETS's superiority over PETS without further in-depth analysis and experimentation. In conclusion, we express our intent to explore the development of \ellone-MBRL that can effectively work alongside probabilistic models, recognizing the potential advantages of both approaches.

\end{document}